\documentclass[preprint,12pt]{elsarticle}
\usepackage{amsmath}
\usepackage{amssymb}
\usepackage{epsfig}
\usepackage{epstopdf}
\usepackage{graphicx}
\usepackage{indentfirst}
\usepackage{subfigure}
\usepackage{enumerate}
\usepackage{srcltx}
\usepackage{geometry}
\usepackage[ruled, lined, linesnumbered]{algorithm2e}
\usepackage{color}
\usepackage{multirow,mathrsfs}
\usepackage{natbib}
\usepackage{tabularx}
\usepackage{array}
\usepackage{ltablex,booktabs}
\usepackage{hyperref}
\newtheorem{theorem}{Theorem}
\newtheorem{assumption}{Assumption}
\newtheorem{remark}{Remark}

\newtheorem{problem}{Problem}

\newenvironment{proof}{{\it Proof:}}{\hspace*{\fill}~{\mbox
		{\rule[0pt]{1.5ex}{1.5ex}}}\par\endtrivlist\unskip}
\allowdisplaybreaks

\journal{Journal of The Franklin Institute}

\begin{document}

\begin{frontmatter}

\title{Distributed goal assignment strategy for improving leader-following formation control performance}

\author[1]{Yun Ho Choi}
\author[1]{Doik Kim\corref{2}}
\ead{doikkim@kist.re.kr}

\address[1]{Center for Intelligent \& Interactive Robotics
	Research, Korea Institute of Science and Technology, 5, Hwarang-ro 14-gil, Seongbuk-gu, Seoul, 02792, South Korea}

\cortext[2]{Corresponding author.}

\begin{abstract}
This paper investigates a distributed goal assignment problem in leader-following formation control of second-order multi-agent systems. It is assumed that each agent can communicate with nearby agents within the communication range and the leader information is only available to a subset of agents.
Compared with existing formation control schemes addressing the goal assignment issue, the main contribution of this paper is to construct a novel distributed assignment strategy allotting appropriate goal positions of agents in the leader-following formation control framework.
Based on the rigorous analysis using the Lyapunov stability theory, the enhancement of the control performance is proved via the proposed assignment strategy. To demonstrate the effectiveness of our theoretical results, two examples including multiple quadrotors are simulated.
\end{abstract}


\end{frontmatter}


 \section{Introduction}
\label{sec:intro}
Distributed coordination of multi-agent systems
has attracted considerable attention from multidisciplinary areas
owing to its broad applications such as scheduling of automated highway systems, attitude alignment of satellites, cooperative transport, and so on \cite{Saber2007,Ren2007}.
One of the fundamental problems of coordination of multi-agent systems is
formation control. The aim of formation control is to design an appropriate control algorithm such that agents form a predefined geometrical shape.
There are a lot of formation control results using different methodologies such as behavior-based \cite{Balch1998},
virtual structure \cite{Lewis1997}, and leader-following approaches \cite{Shao2007}.
Among those methodologies, the leader-following approach has been widely employed because of its simplicity and reliability \cite{Shao2007,Shojaei2018,Gao2019,Lashkari2020}.
In this approach, a predefined formation shape is specified by the relative positions between a referenced agent called a leader and other agents called followers.
Therefore, each follower is controlled to keep the pre-specified relative position from the leader.
For a group of tractor-trailer systems, an adaptive leader-following formation controller using neural networks was designed in \cite{Shojaei2018}. 
In \cite{Gao2019} and \cite{Lashkari2020}, leader-following control schemes were developed for autonomous underwater vehicles without the leader velocity, and mobile robots with docking systems considering the battery failure and maneuverability on rough terrain, respectively.
However, these studies \cite{Shojaei2018,Gao2019,Lashkari2020} assumed that the leader communicates with all followers 
which limits their applicability because the real multi-robot systems have a limited communication range.
To address this problem, a number of researchers have developed leader-following formation control schemes in the presence of limited leader information under the distributed communication network.  
In \cite{Shi2019,Kang2020}, formation control and collision avoidance problems were handled simultaneously by adopting the artificial potential function approach.
Recently, time-varying formation trackers were designed considering mismatched disturbances \cite{Hua2020} or switching communication topology \cite{Wang2019}. 
For practical applications of leader-following formation control approaches, various robotic systems were handled such as mobile robots  \cite{Li2021,Yan2020}, underwater vehicles\cite{Yan2019,Liu2021a}, surface vehicles \cite{Liang2021}, and quadrotors \cite{Liu2021b}.   
Nevertheless, in the existing leader-following formation control schemes 
\cite{Shi2019,Kang2020,Hua2020,Wang2019,Li2021,Yan2020,Yan2019,Liu2021a,Liang2021,Liu2021b}, the goal assignment issue allotting appropriate goal positions for followers was not investigated where the goal indicates the desired relative position with respect to the leader.
Namely, the goals of followers in \cite{Shi2019,Kang2020,Hua2020,Wang2019,Li2021,Yan2020,Yan2019,Liu2021a,Liang2021,Liu2021b} are not allowed to be changed which leaves room for the development of assignment strategy in the leader-following formation control field.

In coordinated control of multi-agent systems, goal assignment plays a key role
because it can reduce the time required to reach the goal position. From the practical point of view, this advantage becomes more significant when the scale of the team of agents is huge likes the drone show. 
Thus, substantial efforts have been devoted to resolving the goal assignment issue.
Centralized assignment strategies were suggested to minimize the sum of the distance traveled or the sum of the squared distance traveled by agents
in \cite{Yu2013,Turpin2014}. To improve scalability, distributed goal assignment strategies have been developed recently.
In \cite{Panagou2020}, a multiple Lyapunov function method was used to assign goals of agents as well as to generate safe trajectories in a distributed manner.
In \cite{Wang2020}, a distributed algorithm assigning goals and creating a collision-free path to the goals was presented
where the path is given by a sequence of nodes of the grid-based environment.
Despite this success, the goal positions considered in \cite{Panagou2020,Wang2020} are defined as absolute positions,
and thus the existing assigning methods cannot be applied to the leader-following formation control
where goals of followers are given by relative positions from the leader. 
Moreover, their approaches are only applicable to first-order multi-agent systems.
In \cite{Kowalczyk2019}, a leader-following formation control problem with distributed goal assignment was investigated.
The main procedure of the goal assignment algorithm in \cite{Kowalczyk2019} is to select two agents and exchange their goals
if an error under the current goals is bigger than that under the exchanged goals. 
Recently, this work is extended to second-order systems \cite{Choi2021b}. 
However, there works \cite{Kowalczyk2019,Choi2021b} suffered from the fact that all the followers should know the leader information. Thus, the centralized communication with the center node being the leader for leader-following multi-agent systems should be available. From the practical point of view, the centralized communication is not available for some robots with limited communication range and it can encounter a significant bottleneck on the central node as the number of followers increases.
That is, the strategies in \cite{Kowalczyk2019,Choi2021b} cannot deal with the assignment problem of large-scale multi-robot systems under the distributed communication network.  
There is a recent formation control approach with distributed goal assignment was presented in \cite{Choi2021}. 
However, since this assignment strategy was developed based on the distance-based formation control framework, we cannot apply this strategy to the leader-following formation control framework.

Motivated by these observations, this paper aims to develop a novel distributed goal assignment strategy for leader-following formation control. We assume that the leader information is only available for a fraction of second-order followers and followers can communicate with nearby agents within their communication range.
Compared with existing related literature, the major contribution of this paper is three-fold:

(C1) This is the first attempt to address the goal assignment issue in leader-following formation control with limited leader information. That is, the goals of followers are assigned properly using an online goal assignment strategy different from the previous leader-following formation control schemes with fixed goals \cite{Shi2019,Kang2020,Hua2020,Wang2019,Li2021,Yan2020,Yan2019,Liu2021a,Liang2021,Liu2021b}.

(C2) In contrast to the centralized communication-based control and assignment scheme \cite{Kowalczyk2019,Choi2021b}, the proposed controller and assignment algorithm can be performed under the distributed communication network by designing a distributed leader estimator.

(C3) To reflect the distributed communication in goal assignment, a novel assignment strategy exchanging not only the goal positions but also the neighbors is developed. Besides, it is rigorously proved that the formation control performance is improved using the proposed assignment strategy.

The remainder of this paper is structured as follows.
In Section \ref{sec:formulation}, some preliminaries are given and the control problem is formulated.
In Section \ref{sec:estimator-controller}, a distributed estimator and a controller are derived and the asymptotic stability is analyzed.
Section \ref{sec:goalassignment} introduces a new goal assignment algorithm and the improved control performance is revealed using the proposed assignment algorithm.
To illustrate the feasibility of our theoretical results, a group of multiple quadrotors are simulated in Section \ref{sec:simulation}.
 Conclusions are drawn in Section \ref{sec:conclusion}.

\section{Preliminaries and problem formulation}
\label{sec:formulation}

\subsection{Second-order multi-agent systems}
We consider the formation control problem of a group of agents in $\mathbb{R}^{d}$
consisting of a virtual leader and $N$ followers with $d=2,3$.
Each follower is modeled as follows:
\begin{align}
\begin{array}{l}
\dot{p}_{i}(t) = v_{i}(t)\\
\dot{v}_{i}(t) = u_{i}(t)
\end{array}	
\label{eq:agent}
\end{align}
where $i=1,\dots,N$, $p_{i}(t)$,  $v_{i}(t)$, and $u_{i}(t)$ denote the position, velocity, and control input of agent $i$, respectively. In this paper, the leader's position is defined as $p_{0}(t)$ which is independent of followers.

\begin{assumption}
\emph{
	Let the velocity and acceleration of the leader be $v_{0}(t) \triangleq \dot{p}_{0}(t)$ and $
	u_{0}(t) \triangleq \ddot{p}_{0}(t)$. 
	Then, the acceleration and its derivative is bounded as $\|u_{0}(t)\| \leq C_{u,0}$ and  $\|\dot{u}_{0}(t)\| \leq C_{u,1}$
	where $C_{u,0}$ and $C_{u,1}$ are positive constants.
}
	\label{as:leader-bound}
\end{assumption}

\subsection{Communication environment}
This section introduces the communication setup for the considered multi-agent system \eqref{eq:agent}.
\begin{assumption}
\emph{
	\label{as:follower-communication}
	Followers can communicate with nearby followers that are within their communication range $R$.}
\end{assumption}

\begin{assumption}
\emph{Leader information $p_{0}(t)$, $v_{0}(t)$, and $u_{0}(t)$ are only available for a fraction of followers.}
\label{as:leader-communication}
\end{assumption}
%

To represent the communication among agents, we consider a graph $\mathcal{G}(\mathcal{V}, \mathcal{E})$
where $\mathcal{V} = \{0,1,\dots,N\}$ and $\mathcal{E} \subseteq \mathcal{V} \times \mathcal{V}$; node $0$ indicates the leader and nodes $1,\dots,N$ denote the followers. There exist undirected edges $(i,j) \in \mathcal{E}$ if $\|p_{i}(t) - p_{j}(t)\| \leq R$ where $p_{i}(t)$ and $p_{j}(t)$ denote the positions of follower $i$ and follower $j$, respectively, and $j \neq i$. Different from the edges among followers, the edges from the leader to a subset of the followers are regarded as directed edges because the leader has an independent motion.

In this paper, a distinct graph $\bar{\mathcal{G}}(\mathcal{V},\bar{\mathcal{E}})$ is introduced to characterize
the controller network where $\bar{\mathcal{E}} \subseteq \mathcal{E}$ \cite{Franchi2012}.
This graph defines the set of information required in the local control system. 
Since the communication among agents is range-limited, $\bar{\mathcal{G}}(\mathcal{V},\bar{\mathcal{E}})$ is a subgraph of  $\mathcal{G}(\mathcal{V}, \mathcal{E})$. 
That is, $\bar{\mathcal{E}} \subseteq \mathcal{E}$. 

To describe the controller network among followers, we consider an undirected graph $\bar{\mathcal{G}}_{F}(\mathcal{V}_{F},\bar{\mathcal{E}}_{F})$ where $\mathcal{V}_{F} = \{1,\dots,N\}$ and $\bar{\mathcal{E}}_{F} \subset \bar{\mathcal{E}}$  Let an adjacency matrix be $A_{F} = [a_{ij}] \in \mathbb{R}^{N \times N}$ for $\bar{\mathcal{G}}_{F}$
and a leader adjacency matrix be $B = \mathrm{diag}\{b_{1},\dots,b_{N}\} \in \mathbb{R}^{N \times N}$ associated with $\bar{\mathcal{G}}$ where $a_{ii} = 0$, $a_{ij} = a_{ji}$, $a_{ij} = 1$ if $(j,i) \in \bar{\mathcal{E}}_{F}$, $a_{ij} = 0$ otherwise, and $b_{i} = 1$ if $(0,i) \in \bar{\mathcal{E}}$ and $b_{i} = 0$ otherwise. 
From $A_{F}$, the Laplacian matrix $L_{F} = [l_{ij}] \in \mathbb{R}^{N \times N} $ is defined as $l_{ii} = \sum_{j \neq i}a_{ij}$ and $l_{ij} = -a_{ij}$ where $j \neq i$. 
According to the two graphs $\mathcal{G}(\mathcal{V},\mathcal{E})$ and $\bar{\mathcal{G}}(\mathcal{V},\bar{\mathcal{E}})$,
the set of neighbors $\mathcal{N}_{i}$ and $\bar{\mathcal{N}}_{i}$ are defined as $\mathcal{N}_{i} = \{j|~ (i,j) \in \mathcal{E} \}$ and $\bar{\mathcal{N}}_{i} = \{j|~ (i,j) \in \bar{\mathcal{E}} \}$, respectively where $i= 1,\dots,N$.
For the control graph $\bar{\mathcal{G}}$, we use the following assumption.

\begin{assumption}
\emph{	\cite{Zhang2015}
		The graph $\bar{\mathcal{G}}$ has a spanning tree with the root node being the leader. 
		That is, there exist paths from the leader to all followers where each path is defined as a sequence of ordered edges of the form $(0 , i_{1} ), (i_{1} , i_{2} ), \dots,(i_{k} , i)$ with $i=1,\dots,N$.
		\label{as:spanning}
	}
\end{assumption}


\begin{figure}
	\centering
	\subfigure[]{\epsfig{figure=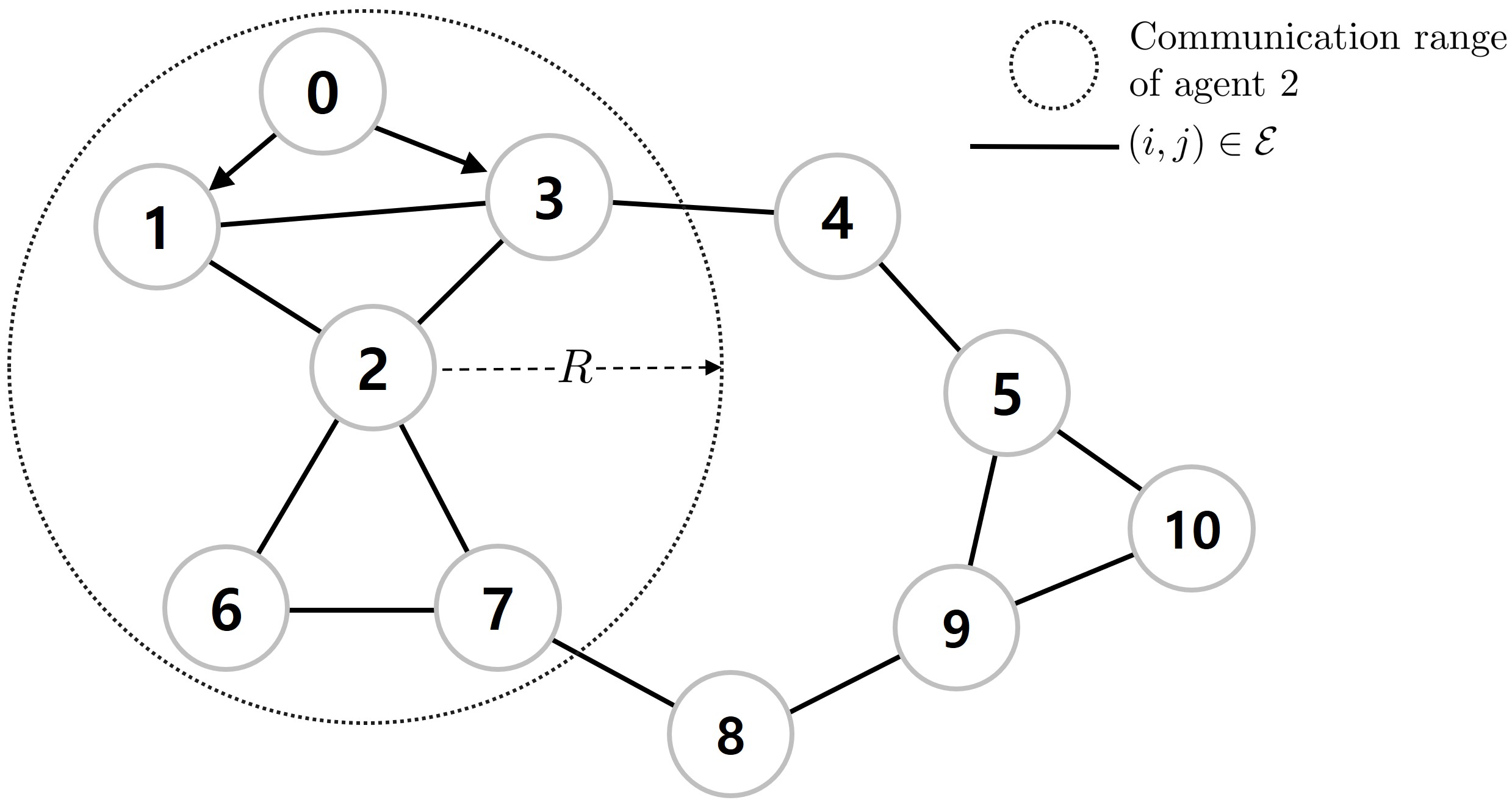, width=7cm,height=4cm}}
	\subfigure[]{\epsfig{figure=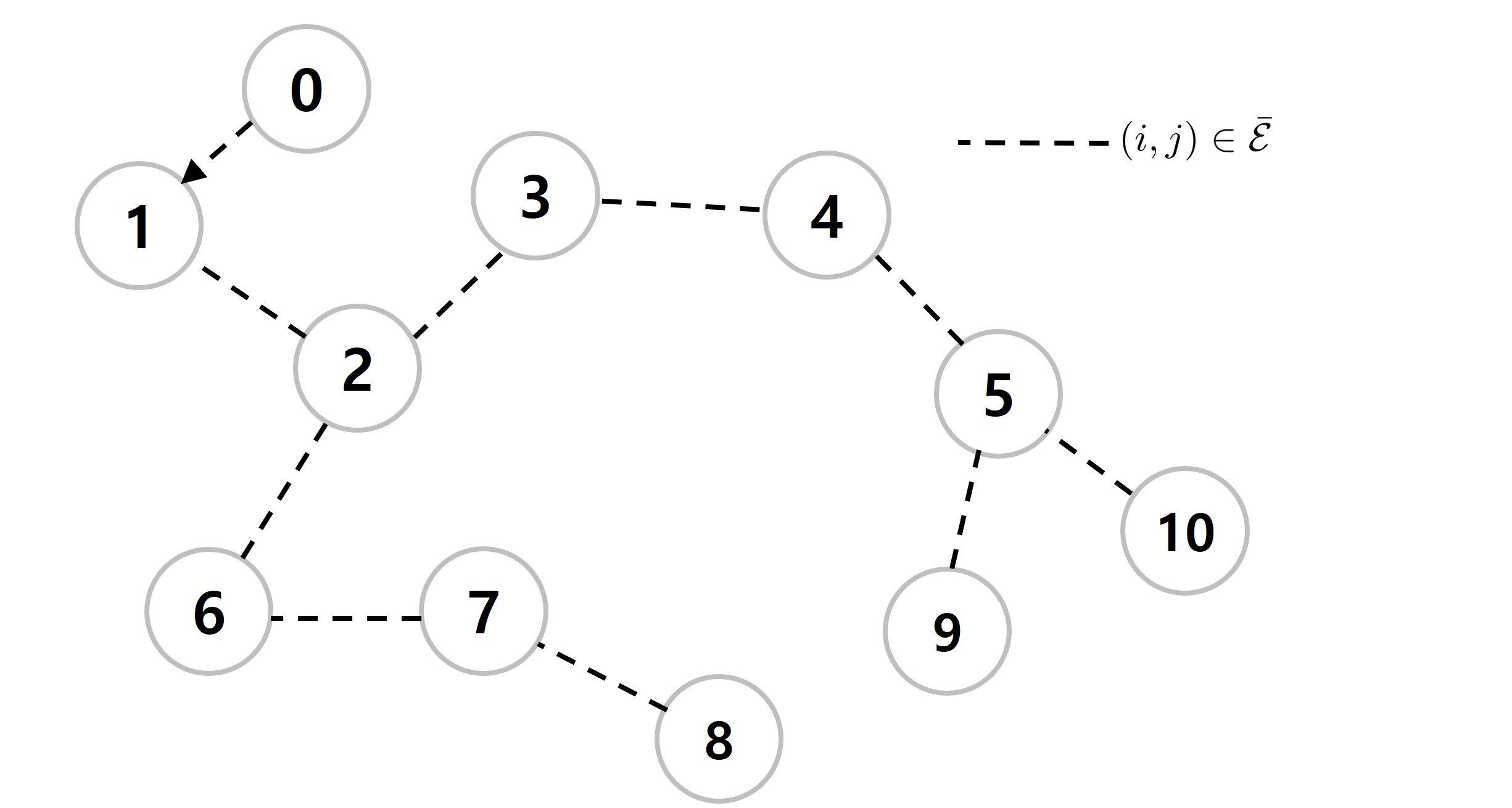, width=7cm,height=4cm}}
	\caption{Example of two network topologies
		(a) $\mathcal{G}(\mathcal{V},\mathcal{E})$
		(b) $\bar{\mathcal{G}}(\mathcal{V},\bar{\mathcal{E}})$.}
	\label{fig:ex-graph}
\end{figure}

Fig. \ref{fig:ex-graph}(a) and \ref{fig:ex-graph}(b) illustrate an example of the communication graph $\mathcal{G}(\mathcal{V}, \mathcal{E})$ 
and the control graph $\bar{\mathcal{G}}(\mathcal{V}, \bar{\mathcal{E}})$, respectively, satisfying Assumptions \ref{as:follower-communication}, \ref{as:leader-communication}, and \ref{as:spanning}. 
From this figure, it is seen that the control graph in Fig. \ref{fig:ex-graph}(b) is a subgraph of the communication graph in Fig. \ref{fig:ex-graph}(a).

\subsection{Problem formulation}
In this paper, goal positions $G_{i}$, $i=1,\dots,N$, are defined as the relative desired positions with respect to the leader to form the target shape. For the goal positions, we use the following assumption.
\begin{assumption} 
\emph{\cite{Panagou2020} Each agent can converge to any goal positions with respect to the leader. 
	That is, the goal positions are interchangeable.
\label{as:goal}}
\end{assumption}

\begin{problem}
\emph{Consider the leader-following multi-agent systems consisting of one leader and $N$ followers \eqref{eq:agent} 
under Assumptions \ref{as:follower-communication}--\ref{as:goal}.
Our problem is designing distributed control laws to form the target formation and developing an online goal assignment strategy exchanging goals of the agents to improve the leader-following formation control performance.}
\label{pro:1}
\end{problem}


\begin{remark}
	\emph{
	Different from the existing distributed leader-following formation control schemes \cite{Shi2019,Kang2020,Hua2020,Wang2019,Li2021,Yan2020,Yan2019,Liu2021a,Liang2021,Liu2021b} with fixed goal positions,
	we consider interchangeable goals as stated in Assumption \ref{as:goal}.
	In \cite{Panagou2020,Wang2020}, distributed goal assignment algorithms were presented for multi-agent systems,
	but those cannot be used for Problem \ref{pro:1} because there was no leader and only the first-order agents were handled. 
	Even though leader-following control schemes with goal assignment were proposed in \cite{Kowalczyk2019,Choi2021b},
	these necessitate the network connection from the leader to all followers.
	In conclusion, Problem \ref{pro:1} is not solved by the previous related literature \cite{Shi2019,Kang2020,Hua2020,Wang2019,Li2021,Yan2020,Yan2019,Liu2021a,Liang2021,Liu2021b,Panagou2020,Wang2020,Kowalczyk2019,Choi2021b}.}
\end{remark}


\section{Distributed formation controller Design}
\label{sec:estimator-controller}
\subsection{Leader estimator design}
\label{sec:estimator}
Motivated by the work of Liang et al. \cite{Liang2019}, a distributed leader estimator for each follower is derived as follows:
\begin{align}
\begin{array}{l}
\dot{\hat{p}}_{i,0} = \hat{v}_{i,0} \\
\dot{\hat{v}}_{i,0} = \hat{u}_{i,0} \\
\dot{\hat{u}}_{i,0} = -\gamma_{i,1} \bigg(\sum_{j=1}^{N} \bar{a}_{F,ij}(\hat{p}_{i,0}-\hat{p}_{j,0}) 
+ \bar{b}_{i}(\hat{p}_{i,0} - p_{0}) \bigg)
\\ \qquad ~ -\gamma_{i,2} \bigg(\sum_{j=1}^{N} \bar{a}_{F,ij}(\hat{v}_{i,0} - \hat{v}_{j,0}) + \bar{b}_{i}(\hat{v}_{i,0} - v_{0}) \bigg)
\\ \qquad ~ -\gamma_{i,3} \bigg(\sum_{j=1}^{N} \bar{a}_{F,ij}(\hat{u}_{i,0} - \hat{u}_{j,0}) + \bar{b}_{i}(\hat{u}_{i,0} - u_{0}) \bigg)
\end{array}
\label{eq:estimator}
\end{align}
where $i=1,\dots,N$, $\gamma_{i,1}, \gamma_{i,2}, \gamma_{i,3}>0$ are positive constants denoting the estimator gains, 
and $\hat{p}_{i,0}$, $\hat{v}_{i,0}$, and $\hat{u}_{i,0}$ are the estimated leader position, velocity, and acceleration, respectively.

Let us define local estimation errors $\tilde{p}_{i,0} = \hat{p}_{i,0} - p_{0}$, $\tilde{v}_{i,0} = \hat{v}_{i,0} - v_{0}$, 
and $\tilde{u}_{i,0} = \hat{u}_{i,0} - u_{0}$. 
Then, we get
\begin{align}
\begin{array}{l}
\dot{\tilde{p}}_{0} = \tilde{v}_{0}\\
\dot{\tilde{v}}_{0} = \tilde{u}_{0}\\
\dot{\tilde{u}}_{0} = -\gamma_{1}((\bar{L}_{F}+\bar{B}) \otimes I_{d})\tilde{p}_{0}
-\gamma_{2}((\bar{L}_{F}+\bar{B}) \otimes I_{d})\tilde{v}_{0}
\\ \qquad
-\gamma_{3}((\bar{L}_{F}+\bar{B}) \otimes I_{d})\tilde{u}_{0} - 1_{N} \otimes \dot{u}_{0}
\end{array}
\label{eq:estimator-error}
\end{align}
where
$\tilde{p}_{0} = [\tilde{p}_{1,0}^{\top},\dots,\tilde{p}_{N,0}^{\top}]^{\top} \in \mathbb{R}^{dN}$,
$\tilde{v}_{0} = [\tilde{v}_{1,0}^{\top},\dots,\tilde{v}_{N,0}^{\top}]^{\top} \in \mathbb{R}^{dN}$,
$\tilde{u}_{0} = [\tilde{u}_{1,0}^{\top},\dots,\tilde{u}_{N,0}^{\top}]^{\top} \in \mathbb{R}^{dN}$,
$I_{d}$ is the $d \times d$ identity matrix, 
$\bar{L}_{F} \in \mathbb{R}^{N \times N}$ is the Laplacian matrix associated with $\bar{\mathcal{G}}_{F}$,
$\bar{B} = \mathrm{diag}\{\bar{b}_{1},\dots,\bar{b}_{N} \}$, $1_{N}$ is $N$-vector of ones, $\gamma_{j} = \mathrm{diag}\{ \gamma_{j,1}, \dots,\gamma_{j,N}\}$ with $j=1,2,3,$, and $\otimes$ indicates the Kronecker product.

Consider a vector $q = [\tilde{p}_{0}^{\top}, \tilde{v}_{0}^{\top},  \tilde{u}_{0}^{\top}]^{\top} \in \mathbb{R}^{3dN}$ and 
$H \triangleq (\bar{L}_{F}+\bar{B}) \otimes I_{d \times d}$.
Then, the time derivative of $\dot{q}$ using \eqref{eq:estimator-error} is expressed by
\begin{align}
\dot{q} = A_{1}q + A_{2}
\label{eq:estimator-error-vec}
\end{align}
where
\begin{align}
\nonumber
A_{1} & =
\left[
\begin{array}{ccc}
0_{dN \times dN} & I_{dN \times dN} & 0_{dN \times dN}\\
0_{dN \times dN} & 0_{dN \times dN} & I_{dN \times dN}\\
-\gamma_{1}H &  -\gamma_{2}H &  -\gamma_{3}H
\end{array}
\right] \in \mathbb{R}^{3dN \times 3dN}, ~~
A_{2} & =
\left[
\begin{array}{c}
0_{dN} \\
0_{dN} \\
-1_{N} \otimes \dot{u}_{0}
\end{array}
\right]  \in \mathbb{R}^{3dN}
\end{align}
Owing to Assumption \ref{as:spanning}, $H$ is positive definite \cite{Zhang2015}, and thus $A_{1}$ is a stable matrix since $\gamma_{i}$, $i=1,2,3$, are positive definite. Therefore, for any matrix $Q>0$, there exists a positive definite matrix $P >0$ satisfying Lyapunov equation
$A_{1}^{\top}P + P A_{1} = -Q$. Following the procedure of Section 3.3 in \cite{Liang2019} and using the practical stability concept, we can easily ensure the uniformly ultimately boundedness of the estimation errors. Note that the estimator \eqref{eq:estimator} can be replaced by other estimators presented in cooperative control problems of multi-agent systems with a leader such as the finite-time estimator in \cite{Zhao2015} or the fixed-time estimator in \cite{Zuo2019}.

\subsection{Controller design via backstepping technique}
\label{sec:controller}
An auxiliary variable $p_{i}^{*}$ is considered to indicate the goal position assigned to the $i$th follower.
Note that the initial value of $p_{i}^{*}$ is set to $G_{i}$ (i.e., $p_{i}^{*}(0) = G_{i}$) for $i=1,\dots,N$ but it could be varied along the goal assignment mechanism which will be shown in the next section. 
To design the local controller, we use the backstepping technique \cite{Krstic1995}. 
Then, two control error surfaces using the estimate leader position are defined as
\begin{align}
\begin{array}{l}
e_{i,1}  = p_{i} - \hat{p}_{i,0} - p_{i}^{*} \\
e_{i,2}  = v_{i} - \zeta_{i}
\end{array}
\label{eq:error_surface}
\end{align}
where $i=1,\dots,N$ and $\zeta_{i}$ indicates the virtual controller of the $i$th follower.

According to the recursive design procedure of backstepping technique, the local virtual and actual controllers are derived as follows:
\begin{align}
\zeta_{i} &= -k_{i,1}e_{i,1} + \hat{v}_{i,0}
\label{eq:virtual_control}
\\
u_{i} &= -k_{i,2}e_{i,2}  - e_{i,1} + \dot{\zeta}_{i}
\label{eq:actual_control}
\end{align}
where $k_{i,1}, k_{i,2}>0$ are control gains.

To analyze the closed-loop stability, let us consider a Lyapunov function as follows: 
\begin{align}
V = \frac{1}{2}e_{1}^{\top}e_{1} + \frac{1}{2}e_{2}^{\top}e_{2}
\label{eq:V}
\end{align}
with $e_{1} = [e_{1,1}^{\top},\dots,e_{N,1}^{\top}]^{\top}$ and $e_{2} = [e_{1,2}^{\top},\dots,e_{N,2}^{\top}]^{\top}$. 
On the other hand, applying \eqref{eq:virtual_control} and \eqref{eq:actual_control} into the derivatives of $e_{i,1}$ and $e_{i,2}$, we get
\begin{align}
\begin{array}{l}
\dot{e}_{i,1}  = -k_{i,1} e_{i,1} + e_{i,2}, \\
\dot{e}_{i,2}  = -k_{i,2} e_{i,2} - e_{i,1}.
\end{array}
\label{eq:dot-error_surface}
\end{align}
Then, the time derivative of $V$ is derived as follows:
\begin{align}
\dot{V} = -e_{1}^{\top}k_{1}e_{1} - e_{2}^{\top}k_{2}e_{2}
\label{eq:dot-V}
\end{align}
where $k_{1} = \mathrm{diag}\{k_{1,1},\dots,k_{N,1} \} \otimes I_{d}$ and $k_{2} = \mathrm{diag}\{k_{1,2},\dots,k_{N,2} \}\otimes I_{d}$.
By letting $k_{m} \triangleq \min_{i=1,\dots,N, j=1,2}\{k_{i,j}\}$, the inequality $\dot{V} \leq -2k_{m}V$ holds. 
Consequently, the asymptotic convergence of the error surfaces $e_{i,1}$ and $e_{i,2}$, $i=1,\dots,N$, is ensured based on the Lyapunov stability theory.
Now, let us define the global formation error $\delta_{i} = p_{i} - p_{0} - p_{i}^{*}$ of the $i$th follower.
From the definition of $e_{i,1}$ and $\tilde{p}_{i,0}$, we get  $\delta_{i} = e_{i,1} + \tilde{p}_{i,0}$. 
Since $\tilde{p}_{i,0}$ is uniformly bounded from Section \ref{sec:estimator} and $e_{i,1}$ converges to zero, 
the global formation error $\delta_{i}$ converges to nearby zero regarding to the magnitude of $\tilde{p}_{i,0}$. Therefore, we can conclude that the desired formation shape is achieved practically. 

\begin{remark}
\emph{
The guideline of the selecting the design parameters of the estimator and the controller is listed as follows:
(i) increasing the estimator gains $\gamma_{i,1}$, $\gamma_{i,2}$, $\gamma_{i,3}$, $i=1,\dots,N$, helps to increase the minimum eigenvalue of $Q$, which reduces the ultimate boundedness of the local estimation errors; (ii) increasing $k_{i,1}$ and $k_{i,2}$ improves the convergence speed of the local formation errors $e_{i,1}$.}
\label{re:gain-selection}
\end{remark}

\section{Distributed goal assignment strategy}
\label{sec:goalassignment}
In order to improve the leader-following formation control performance, this section investigates an online distributed goal assignment problem by assigning proper goals for followers. 
In \cite{Kowalczyk2019,Choi2021b}, goal assignment strategies were suggested
where the goal positions of two selected agents are swapped whenever a sum of the formation errors decreases along with the goal exchange. However, the assignment procedure in \cite{Kowalczyk2019,Choi2021b} cannot be applied to our formation control problem
because it requires that all followers can access the information of the leader.
This condition is available only if the leader communicates with all followers or the leader information is pre-programmed in all followers. However, the pre-programmed approach is of low flexibility to the update of the leader signal in practice. 
Therefore, the previous strategies presented in \cite{Kowalczyk2019,Choi2021b} were based on the centralized communication with the center node being the leader as displayed in Fig. \ref{fig:compare-ctrl-graph}(a).
In contrast, we consider the distributed communication network determined by the communication range and the limited leader information as shown in  \ref{fig:compare-ctrl-graph}(b), which is common in the conventional leader-following formation control schemes \cite{Shi2019,Kang2020,Hua2020,Wang2019,Li2021,Yan2020,Yan2019,Liu2021a,Liang2021,Liu2021b}.
The distributed network in Fig. \ref{fig:compare-ctrl-graph}(b)
is more preferable since centralized communication can encounter a significant bottleneck on the central node
as the number of agents increases. Recently, a distributed goal assignment approach was presented in \cite{Choi2021} but its approach is developed in the sense of distance-based formation control and only the first-order agents were handled. 

\begin{figure}
\centering
\subfigure[]{\epsfig{figure=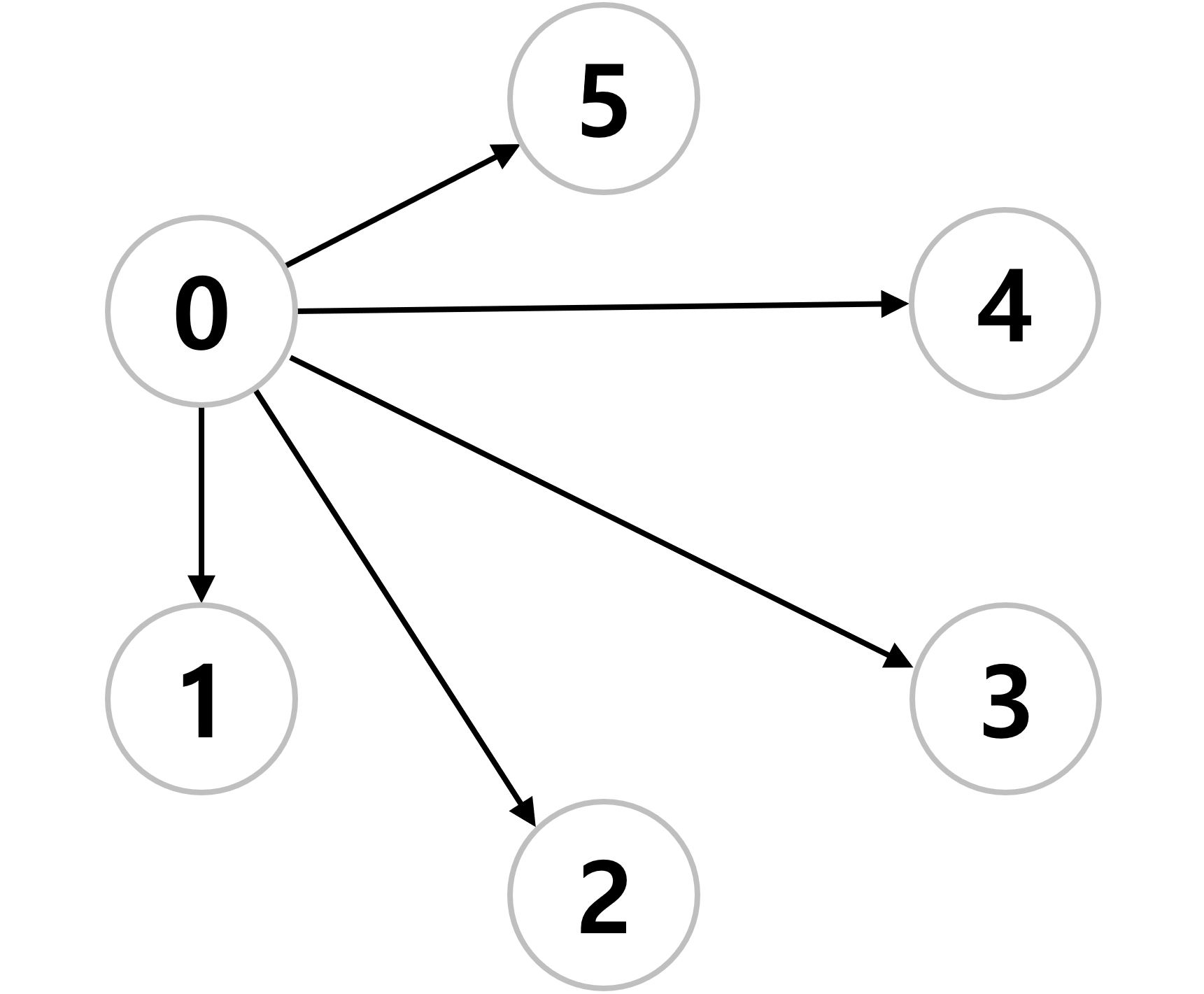, width=4.8cm,height=4cm}}
\subfigure[]{\epsfig{figure=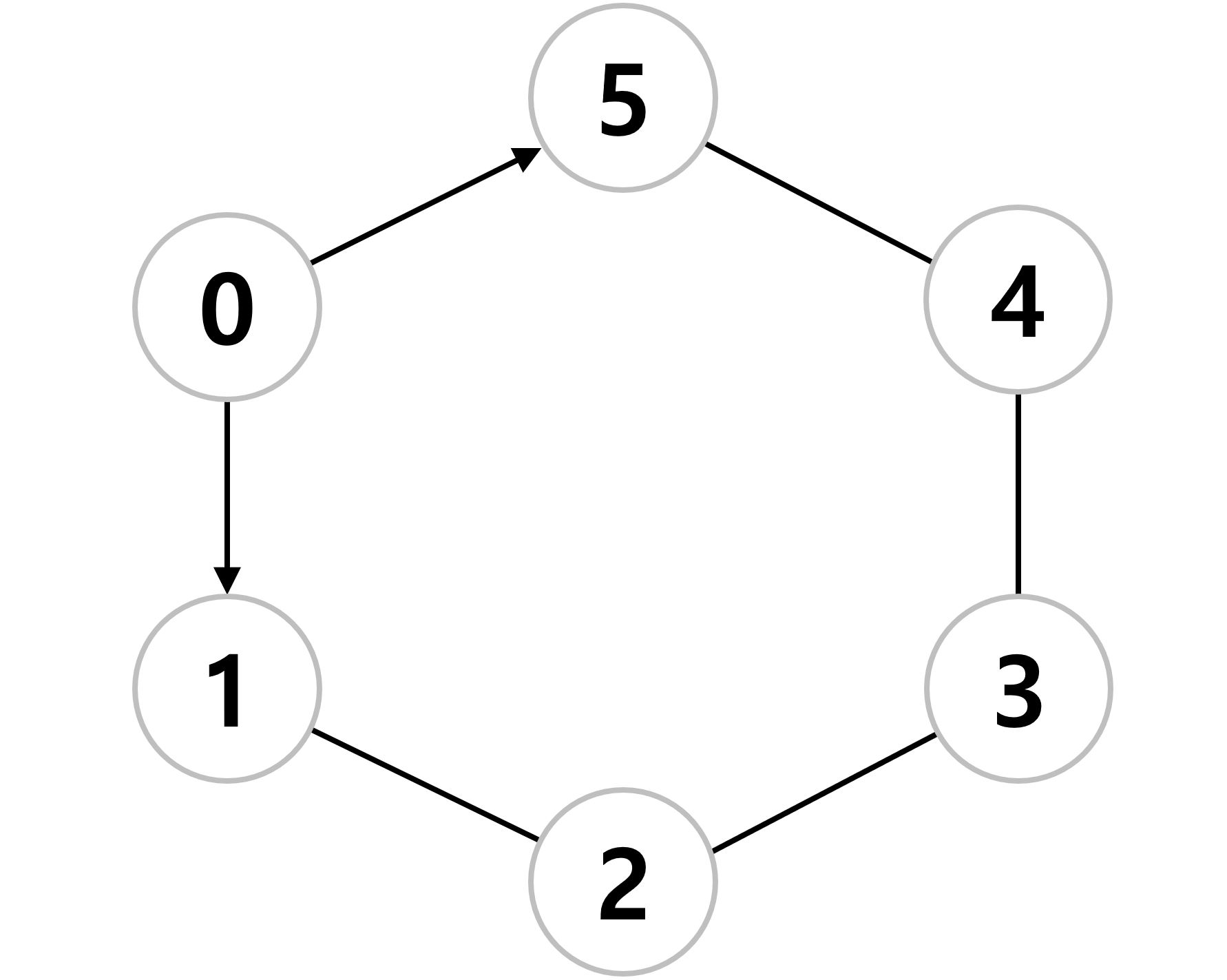, width=4.8cm,height=4cm}}
\caption{Comparison of the communication graph for the same hexagon one leader and five followers
	(a) $\mathcal{G}$ in \cite{Kowalczyk2019}
	(b) $\mathcal{G}$ in the proposed control scheme.}
\label{fig:compare-ctrl-graph}
\end{figure}

Similar to \cite{Kowalczyk2019,Choi2021b}, in the proposed goal assignment procedure, two agents are selected and their goal positions are exchanged if a condition is satisfied. 
However, since the communication range is limited, we should exchange agents' neighbors $\bar{\mathcal{N}}_{i}$ associated with the control graph $\bar{\mathcal{G}}$ as well as their goal positions. 
The spanning tree condition in Assumption \ref{as:spanning} is referred as a basic assumption in the previous leader-following control schemes \cite{Shi2019,Kang2020,Hua2020,Wang2019,Li2021,Yan2020,Yan2019,Liu2021a,Liang2021,Liu2021b}, but this assumption is normally based on the non-exchangeable goals. 
Therefore, Assumption \ref{as:spanning} may not be guaranteed if agents' goals are exchanged in the presence of the limited communication range. 
To handle this issue, we propose an assignment strategy including the update of the set of their neighbors $\bar{\mathcal{N}}_{i}$ associated with the control graph $\bar{\mathcal{G}}$. 
To represent this transition, we use time-dependent variables $\mathcal{N}_{i}(t)$, $p_{i}^{*}(t)$, and $\bar{\mathcal{N}}_{i}(t)$. Let agents $\alpha$ and $\beta$ be the selected two agents in the assignment strategy. Then, a basic assumption of the proposed assignment strategy is given as follows:

\begin{assumption}
\emph{
	For agents $\alpha$ and $\beta$, it holds that $\bar{\mathcal{N}}_{\alpha}(t)-\{\beta\} \subset \mathcal{N}_{\beta}(t)$ and $\bar{\mathcal{N}}_{\beta}(t)-\{\alpha\} \subset \mathcal{N}_{\alpha}(t)$.}
	\label{as:goal-communication}
\end{assumption}

\begin{remark}
	\emph{
	Assumption \ref{as:goal-communication} means that one of the two selected agents can get the information of not only its neighbors
	but also the other's neighbors defined in the control graph $\bar{\mathcal{G}}(t)$.
	This assumption is mandatory in the proposed assignment process to maintain the spanning tree condition of $\bar{\mathcal{G}}(t)$ (i.e., Assumption \ref{as:spanning}) 
	because the goal exchange of two agents basically leads to the change of 
	$\bar{\mathcal{G}}(t)$ by updating their neighbors $\bar{\mathcal{N}}_{i}(t)$.
	For a better understanding of Assumption \ref{as:goal-communication}, let us see Fig. \ref{fig:ex-graph-assumption} illustrating the communication and control graphs of seven agents.
	Note that Assumption \ref{as:goal-communication} is satisfied for agents $4$ and $5$ because
	\begin{align}
		\begin{array}{l}
			\bar{\mathcal{N}}_{4}(t) = \{1, 5\}, \quad  \mathcal{N}_{4}(t) = \{1,2,3,5\}, \\
			\bar{\mathcal{N}}_{5}(t) = \{2, 4\}, \quad  \mathcal{N}_{5}(t) = \{1,2,4,6 \}.
		\end{array}
	\end{align}
}
\end{remark}

\begin{figure}
	\centering
	\subfigure[]{\epsfig{figure=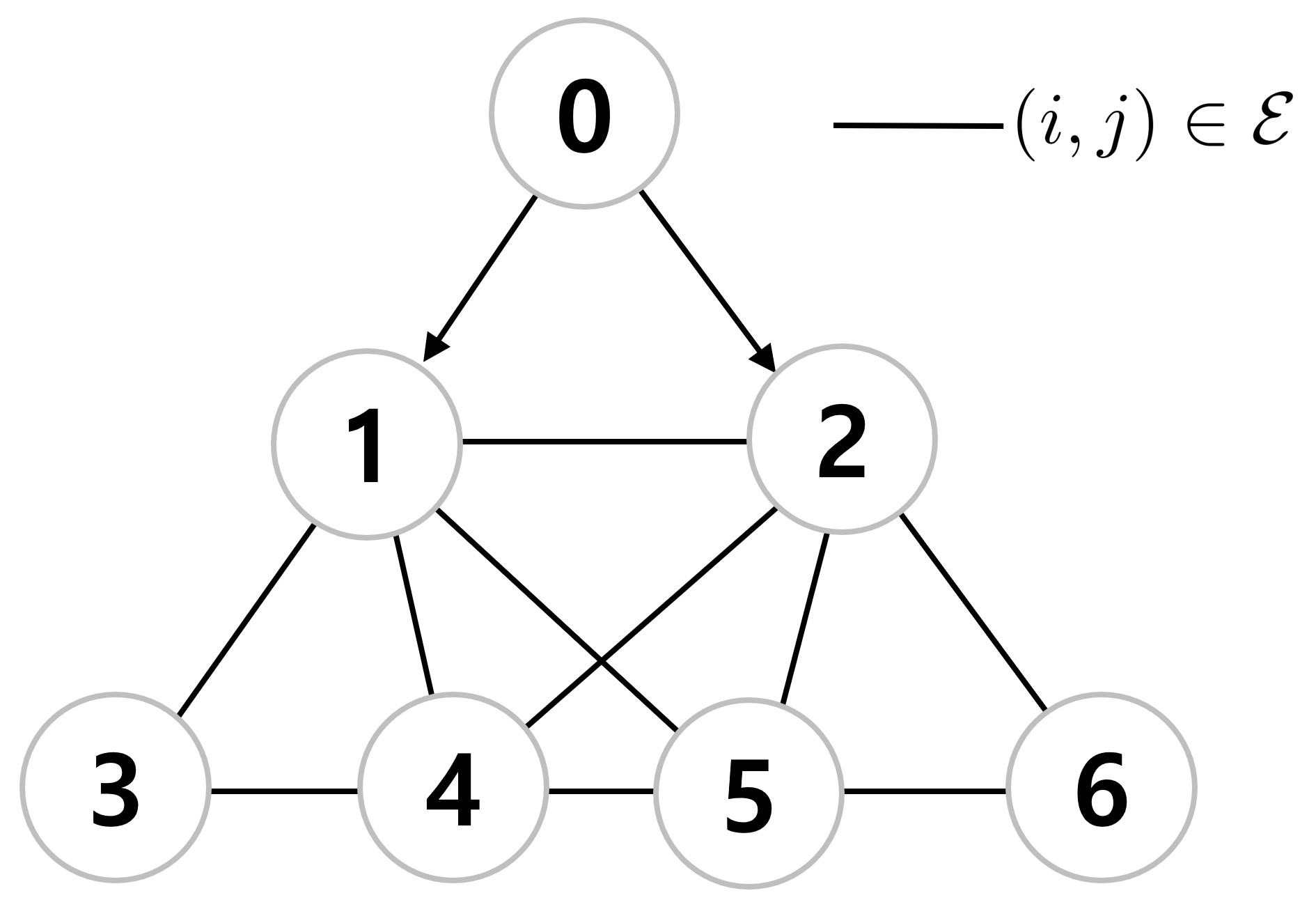, width=5cm,height=3.5cm}}
	\subfigure[]{\epsfig{figure=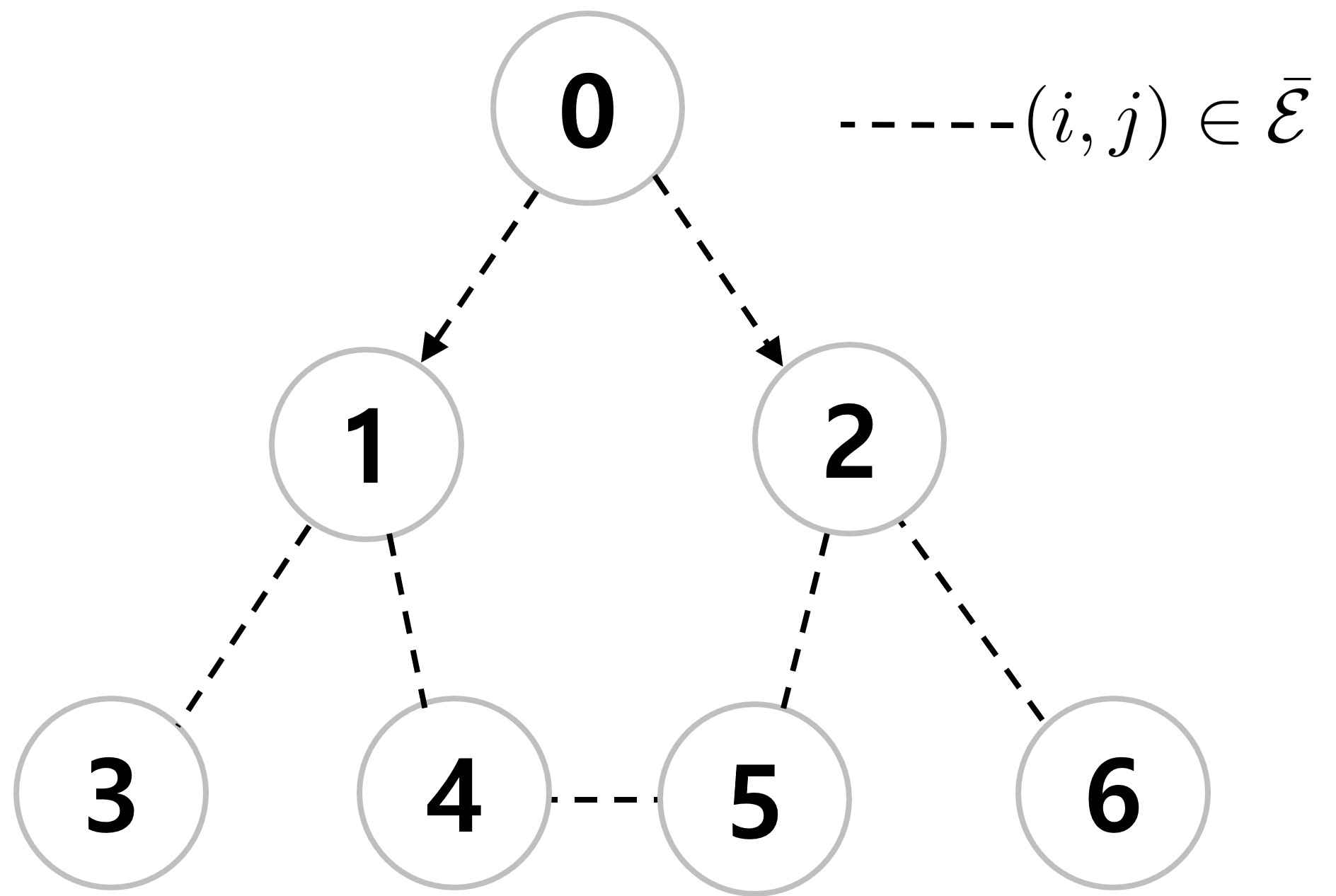, width=5cm,height=3.5cm}}
	\caption{Communication and control graphs satisfying Assumption \ref{as:goal-communication} for $\alpha = 4$ and $\beta = 5$
		(a) $\mathcal{G}(\mathcal{V},\mathcal{E})$
		(b) $\bar{\mathcal{G}}(\mathcal{V},\bar{\mathcal{E}})$.}
	\label{fig:ex-graph-assumption}
\end{figure}

Using Assumption \ref{as:goal-communication}, the proposed distributed online goal assignment algorithm is presented in Algorithm \ref{alg:goal}. Simply, Algorithm \ref{alg:goal} consists of four steps.

\emph{Step 1}: Whenever $t = t_{k}$, a pair of agents is selected and it is monitored whether Assumption \ref{as:goal-communication} is satisfied. Here, $t_{k}$, $k=1,2,\dots$, indicate predetermined periodic instants executing the goal assignment algorithm. 
It should be stressed that $t_{k}$ is selected regardless of the control sampling time \cite{Kowalczyk2019}.

\emph{Step 2:} If the assumption is correct, temporary variables $\breve{\bar{\mathcal{N}}}_{\alpha}$, $\breve{\bar{\mathcal{N}}}_{\beta}$, $\breve{\bar{\mathcal{N}}}_{m}$, $\breve{p}_{\alpha}^{*}$, and $\breve{p}_{\beta}^{*}$ under the exchanged goals of agent $\alpha$ and agent $\beta$ are defined
where $m \in \bar{\mathcal{N}}_{\alpha \cup \beta}(t_{k}^{-})$; $\bar{\mathcal{N}}_{\alpha \cup \beta}(t_{k}^{-}) \triangleq \bar{\mathcal{N}}_{\alpha}(t_{k}^{-}) \cup \bar{\mathcal{N}}_{\beta}(t_{k}^{-}) - \{\alpha\} - \{\beta\}$ and $t_{k}^{-} \triangleq \lim_{\varepsilon \rightarrow 0^{+}}t_{k} - \varepsilon$.
Since the neighbors of agents $\alpha$ and $\beta$ will be updated along with the goal exchange and the communication among agents is undirected, 
it is necessary to take account of $\bar{\mathcal{N}}_{m}$ and $\breve{\bar{\mathcal{N}}}_{m}$. According to the following rules
\begin{align}
\begin{array} {l}
	\breve{\bar{\mathcal{N}}}_{\alpha} = \Big\{
	\begin{array}{l}
		\bar{\mathcal{N}}_{\beta}(t_{k}^{-}) - \{\alpha\} + \{\beta\},~~ \mathrm{if}~ \alpha \in \bar{\mathcal{N}}_{\beta}(t_{k}^{-}), \\
		\bar{\mathcal{N}}_{\beta}(t_{k}^{-}), \qquad \qquad \quad ~~~~ \mathrm{if}~ \alpha \notin \bar{\mathcal{N}}_{\beta}(t_{k}^{-}),
	\end{array}
	\\
	\breve{\bar{\mathcal{N}}}_{\beta} =  \Big\{
	\begin{array}{l}
		\bar{\mathcal{N}}_{\alpha}(t_{k}^{-}) - \{\beta\} + \{\alpha\}, ~~\mathrm{if}~ \beta \in \bar{\mathcal{N}}_{\alpha}(t_{k}^{-}), \\
		\bar{\mathcal{N}}_{\alpha}(t_{k}^{-}), \qquad \qquad \quad ~~~~ \mathrm{if}~ \beta \notin \bar{\mathcal{N}}_{\alpha}(t_{k}^{-}),
	\end{array}
	\\
	\breve{p}_{\alpha}^{*} = p_{\beta}^{*}(t_{k}^{-}), \quad  \breve{p}_{\beta}^{*} = p_{\alpha}^{*}(t_{k}^{-}),
\end{array}
\label{eq:Ni-p-breve}
\end{align}
we get $\breve{\bar{\mathcal{N}}}_{\alpha}$, $\breve{\bar{\mathcal{N}}}_{\beta}$, $\breve{p}_{\alpha}^{*}$, and $\breve{p}_{\beta}^{*}$. From $\breve{\bar{\mathcal{N}}}_{\alpha}$ and $\breve{\bar{\mathcal{N}}}_{\beta}$, $\breve{\bar{\mathcal{N}}}_{m}$ are determined using the bidirectional edges in 
the control graph $\bar{\mathcal{G}}(t)$.

\emph{Step 3:} Two compounded errors $e_{cur}$ and $e_{new}$ are computed where $e_{cur}$ indicates a sum of the local errors under the current goals $p_{\alpha}^{*}$ and $p_{\beta}^{*}$ and $e_{new}$ denotes a sum of the errors under the new goals $\breve{p}_{\alpha}^{*}$ and $\breve{p}_{\beta}^{*}$.
Here, $\breve{e}_{\alpha,j}$ and $\breve{e}_{\beta,j}$ with $j=1,2$ are defined as
\begin{align}
\nonumber
\begin{array}{l}
	\breve{e}_{\alpha,1} = p_{\alpha}(t_{k}^{-}) - \hat{p}_{\alpha,0}(t_{k}^{-}) - \breve{p}_{\alpha}^{*}, \quad 
	\breve{e}_{\alpha,2} = v_{\alpha}(t_{k}^{-}) - \breve{\zeta}_{\alpha}, \\
	\breve{e}_{\beta,1} = p_{\beta}(t_{k}^{-}) - \hat{p}_{\beta,0}(t_{k}^{-}) - \breve{p}_{\beta}^{*},
	\quad
	\breve{e}_{\beta,2} = v_{\beta}(t_{k}^{-}) - \breve{\zeta}_{\beta},
\end{array}
\end{align}
where 
\begin{align}
	\nonumber 
	\begin{array}{l}
	\breve{\zeta}_{\alpha} = -k_{\alpha,1}\breve{e}_{\alpha,1} + \hat{v}_{\alpha,0}(t_{k}^{-}), \quad
	\breve{\zeta}_{\beta} = -k_{\beta,1}\breve{e}_{\beta,1} + \hat{v}_{\beta,0}(t_{k}^{-}).
	\end{array}
\end{align}

\emph{Step 4:} 
If the exchanging condition $e_{cur} > e_{new}$ is satisfied, the two goals are swapped by updating
$\bar{\mathcal{N}}_{\alpha}(t)$, $\bar{\mathcal{N}}_{\beta}(t)$, $\bar{\mathcal{N}}_{m}(t)$, $p_{\alpha}^{*}(t)$ and  $p_{\beta}^{*}(t)$ as the temporary variables
$\breve{\bar{\mathcal{N}}}_{\alpha}$, $\breve{\bar{\mathcal{N}}}_{\beta}$, $\breve{\bar{\mathcal{N}}}_{m}$, $\breve{p}^{*}_{\alpha}$, and $\breve{p}^{*}_{\beta}$, respectively and $t_{k}$ is recorded as a new exchanging moment $\tau_{g}$. There updated information holds until the next goal exchange occurs.

\begin{algorithm}
\SetKwInOut{Input}{input}
\SetKwInOut{Output}{output}
\caption{Distributed Goal Assignment Algorithm}
\label{alg:goal}
\Input{$t_{k}$, $\alpha$, $\beta$, $\bar{\mathcal{N}}_{\alpha}(t_{k}^{-})$, $\bar{\mathcal{N}}_{\beta}(t_{k}^{-})$, $\bar{\mathcal{N}}_{m}(t_{k}^{-})$, $p_{\alpha}^{*}(t_{k}^{-})$, $p_{\beta}^{*}(t_{k}^{-})$ ($m \in \bar{\mathcal{N}}_{\alpha \cup \beta}(t_{k}^{-}))$}
\Output{$\tau_{g}$, $\bar{\mathcal{N}}_{\alpha}(t)$, $\bar{\mathcal{N}}_{\beta}(t)$, $\bar{\mathcal{N}}_{m}(t)$, $p_{\alpha}^{*}(t)$, $p_{\beta}^{*}(t)$ ($t \geq t_{k}$)}
\If{Assumption \ref{as:goal-communication} holds for agents $\alpha$ and $\beta$ }{
	Define $\breve{\bar{\mathcal{N}}}_{\alpha}$, $\breve{\bar{\mathcal{N}}}_{\beta}$, $\breve{\bar{\mathcal{N}}}_{m}$,
	$\breve{p}_{\alpha}^{*}$, and $\breve{p}_{\beta}^{*}$ under the exchanged goals of agent $\alpha$ and agent $\beta$;\\
	Compute
	$e_{cur} \leftarrow \sum_{j=1}^{2}(\|e_{\alpha,j}(t_{k}^{-})\|^{2} + \| e_{\beta,j}(t_{k}^{-})\|^{2})$
	and
	$e_{new} \leftarrow \sum_{j=1}^{2}(\|\breve{e}_{\alpha,j}\|^{2} + \|\breve{e}_{\beta,j}\|^{2})$; \\
	\If{$e_{cur}>e_{new}$}{	 						
		$\tau_{g} \leftarrow t_{k}$; $g \leftarrow g+1$; $k \leftarrow k+1$; 
		Update $\bar{\mathcal{N}}_{\alpha}(t)$, $\bar{\mathcal{N}}_{\beta}(t)$, $\bar{\mathcal{N}}_{m}(t)$, $p_{\alpha}^{*}(t)$, $p_{\beta}^{*}(t)$ as
		$\bar{\mathcal{N}}_{\alpha}(t) \leftarrow \breve{\bar{\mathcal{N}}}_{\alpha}$,
		$\bar{\mathcal{N}}_{\beta}(t) \leftarrow \breve{\bar{\mathcal{N}}}_{\beta}$,
		$\bar{\mathcal{N}}_{m}(t) \leftarrow \breve{\bar{\mathcal{N}}}_{m}$,
		$p^{*}_{\alpha}(t) \leftarrow \breve{p}^{*}_{\alpha}$,    $p^{*}_{\beta}(t) \leftarrow \breve{p}^{*}_{\beta}$;
	}\Else{$k \leftarrow k+1$; Keep $\bar{\mathcal{N}}_{\alpha}(t)$, $\bar{\mathcal{N}}_{\beta}(t)$, $\bar{\mathcal{N}}_{m}(t)$, $p_{\alpha}^{*}(t)$, $p_{\beta}^{*}(t)$ }
}\Else{$k \leftarrow k+1$; Keep $\bar{\mathcal{N}}_{\alpha}(t)$, $\bar{\mathcal{N}}_{\beta}(t)$, $\bar{\mathcal{N}}_{m}(t)$, $p_{\alpha}^{*}(t)$, $p_{\beta}^{*}(t)$}
\end{algorithm}

 Our main result is summarized by the following theorem.
\begin{theorem}
Consider multi-agent systems controlled by the backstepping control law \eqref{eq:virtual_control} and \eqref{eq:actual_control} with the estimator \eqref{eq:estimator} under Assumptions \ref{as:leader-bound}--\ref{as:goal}.
If there exists a pair of agents satisfying Assumption \ref{as:goal-communication} and there exist $\tau_{g}$ such that $e_{cur}(\tau_{g}) > e_{new}(\tau_{g})$ for $g=1,2,\dots$, the leader-following formation control performance can be improved with properly assigned goal positions according to Algorithm \ref{alg:goal}.
\label{th:final}
\end{theorem}
\begin{proof}
Firstly, let us assume that agents $\alpha$ and $\beta$ swap their goals at $t = \tau_{1}$. 
Then, along their current goals and the new (i.e., exchanged) goals, two different Lyapunov function candidates $V_{cur,1}(t)$ and $V_{new,1}(t)$ are defined as follows:
\begin{align}
	\begin{array}{l}
	V_{cur,1}(t) = \frac{1}{2}\big(\|p_{\alpha}(t) - \hat{p}_{\alpha,0}(t) - p_{\alpha}^{*} \|^{2} 
	+ \|v_{\alpha}(t) - \zeta_{\alpha}(t) \|^{2} 
	\\ 
	\qquad \qquad \quad +\|p_{\beta}(t) - \hat{p}_{\beta,0}(t) - p_{\beta}^{*} \|^{2}  + \|v_{\beta}(t) - \zeta_{\beta}(t) \|^{2}	
	\big) + \bar{V}_{cur,1}(t)
	\\
	V_{new,1}(t) = \frac{1}{2}\big(\|p_{\alpha}(t) - \hat{p}_{\alpha,0}(t) - p_{\beta}^{*} \|^{2} 
	+ \|v_{\alpha}(t) - \zeta_{\alpha}(t) \|^{2} 
	\\ 
	\qquad \qquad \quad + \|p_{\beta}(t) - \hat{p}_{\beta,0}(t) - p_{\alpha}^{*} \|^{2}  + \|v_{\beta}(t) - \zeta_{\beta}(t) \|^{2}	
	\big) + \bar{V}_{new,1}(t)
	\end{array}
\label{eq:V_cur_V_new}
\end{align}
Here, $\bar{V}_{cur,1}(t)$ and $\bar{V}_{new,1}(t)$ indicate the remaining terms. 
Since the continuity of the states (i.e., $p_{i}(t)$ and $v_{i}(t)$) and estimated leader information (i.e., $\hat{p}_{i,0}(t)$, $\hat{v}_{i,0}(t)$, and $\hat{u}_{i,0}(t)$) is preserved when the goals are swapped, it holds that $\bar{V}_{cur,1}(\tau_{1})$ = $\bar{V}_{new,1}(\tau_{1})$. Then, from the definitions of $e_{cur}$ and $e_{new}$ and the exchanging condition in Algorithm \ref{alg:goal},
it leads to $V_{cur,1}(\tau_{1}) > V_{new,1}(\tau_{1})$.

Now, we will show that $V_{cur,1}(t) > V_{new,1}(t)$ for $t \geq \tau_{1}$. 
Even if the goals are exchanged, the dynamics of the error surfaces in \eqref{eq:dot-error_surface} is not changed. 
Thus, the error surfaces under the current goals and those under the new goals exponentially converge to zero with different initial values but with same speed regarding to the control gains. 
Therefore, from \eqref{eq:V_cur_V_new}, $V_{cur,1}(t)$ and $V_{new,1}(t)$ have same convergence speed for $t \geq \tau_{1}$. 
Owing to $V_{cur,1}(\tau_{1}) > V_{new,1}(\tau_{1})$, the inequality $V_{cur,1}(t) > V_{new,1}(t)$ is guaranteed for $t \geq \tau_{1}$. Applying this property recursively, 
we have 
\begin{align}
	V_{cur,g}(t) > V_{new,g}(t), ~~~ t \geq \tau_{g}.
\end{align}
where $g = 1,2,\dots$.

Let the Lyapunov function $V(t)$ be
\begin{align}
	V(t) = \left\{
	\begin{array}{l}
		V_{cur,1}(t), \quad t \in [t_{0}, \tau_{1}) \\
		V_{cur,2}(t), \quad t \in [\tau_{1}, \tau_{2}) \\
		V_{cur,3}(t), \quad t \in [\tau_{2}, \tau_{3})	\\
		~~~~ \vdots  \qquad \qquad ~~ \vdots
	\end{array}
	\right.
	\label{eq:V-piecewise-const}
\end{align}
Note that $V_{cur,g}(t) = V_{new,g-1}(t)$. Then, we can conclude that 
\begin{align}
	V(0) > V(\tau_{1}) > V(\tau_{2}) > \cdots.
\end{align}
This implies that the proposed strategy improves the leader-following formation control performance since 
a discrete jump to a lower value of the Lyapunov function happens whenever the goals are exchanged. 

\end{proof}

\begin{remark}
\label{re:comp}
\emph{
In the previous goal assignment algorithm in \cite{Choi2021b},
the assignment condition is defined as $e_{cur} > e_{new}$ where
\begin{align}
	\begin{array}{l}
		e_{cur} = \|p_{\alpha}(t_{k}^{-}) - p_{0}(t_{k}^{-}) - p_{\alpha}^{*}(t_{k}^{-})\|^{2}
		+ \|v_{\alpha}(t_{k}^{-}) - \zeta_{\alpha}(t_{k}^{-})\|^{2}
		\\
		\qquad  \quad
		+ \|p_{\beta}(t_{k}^{-}) - p_{0}(t_{k}^{-}) - p_{\beta}^{*}(t_{k}^{-})\|^{2}  + \|v_{\beta}(t_{k}^{-}) - \zeta_{\beta}(t_{k}^{-})\|^{2},
		\\
		e_{new} = \|p_{\alpha}(t_{k}^{-}) - p_{0}(t_{k}^{-}) - \breve{p}_{\alpha}^{*}\|^{2}
		+ \|v_{\alpha}(t_{k}^{-}) - \breve{\zeta}_{\alpha}\|^{2}		
		\\
		\qquad  \quad
		+\|p_{\beta}(t_{k}^{-}) - p_{0}(t_{k}^{-}) - \breve{p}_{\beta}^{*}\|^{2}
		+ \|v_{\beta}(t_{k}^{-}) - \breve{\zeta}_{\beta}\|^{2}	.
	\end{array}
\end{align}
Here, $\breve{p}_{\alpha}^{*} = p_{\beta}^{*}(t_{k}^{-})$ and $\breve{p}_{\beta}^{*} = p_{\alpha}^{*}(t_{k}^{-})$.
From the above equations, one can see that the errors depend on the leader's position $p_{0}(t)$ and velocities $v_{0}(t)$ used in the virtual controllers.
Besides, the leader's acceleration $u_{0}(t)$ is also used in the local actual controller of each follower. 
Thus, the implicit assumption in \cite{Choi2021b} is that all followers have access to the leader but it is not valid in a general distributed network (see Fig. \ref{fig:compare-ctrl-graph}(b)). Compared with \cite{Choi2021b}, the pros and cons of the proposed assignment algorithm are summarized as follows.
	\\ \indent
	\emph{Pros:} (i) The proposed assignment algorithm can be implemented for multi-agent systems with limited leader information under distributed network; and (ii)  employing the estimated leader information $\hat{p}_{i,0}(t)$ and $\hat{v}_{i,0}(t)$ 
	from the distributed estimator \eqref{eq:estimator}, a local controller and assignment condition are derived and the enhancement of the distributed control performance is rigorously analyzed in the sense of Lyapunov (see the proof of Theorem \ref{th:final}).
	\\ \indent
	\emph{Cons:} (i) The computation for the estimator \eqref{eq:estimator} and the transmission of the estimated signals (i.e., $\hat{p}_{i,0}$, $\hat{v}_{i,0}$, $\hat{u}_{i,0}$) among agents are necessary; and (ii) to implement the proposed distributed assignment algorithm, a pair of agents needs to known each other's neighbors (i.e., Assumption \ref{as:goal-communication} should be satisfied). 
}
\end{remark}

 \begin{remark}
 	\emph{
	In the recent goal assignment study \cite{Choi2021b}, the structure of the error surfaces are similar to ours and 
	the analysis showing that $V(t)$ decreases whenever the goals are swapped was given similarly (see the proof of Theorem 2 in \cite{Choi2021b}). 
	In spite of these similarities, the goal assignment strategy in \cite{Choi2021b} cannot be extended to solve our problem 
	because the only update of the goals in the assignment procedure in \cite{Choi2021b} may lose the connectivity because of the  limited communication range. 
	In order to preserve the connectivity while exchanging goals, we update not only the goals but also the neighbors in the proposed assignment algorithm (see $\bar{\mathcal{N}}_{\alpha}(t)$, $\bar{\mathcal{N}}_{\beta}(t)$, $\bar{\mathcal{N}}_{m}(t)$ in Algorithm \ref{alg:goal}). 
	This is the main different part from \cite{Choi2021b} in terms of goal assignment.}
	\label{re:GA-neighbor-update}
\end{remark}

\begin{figure}[h]
	\centering
	\includegraphics[width=5cm,height=2cm]{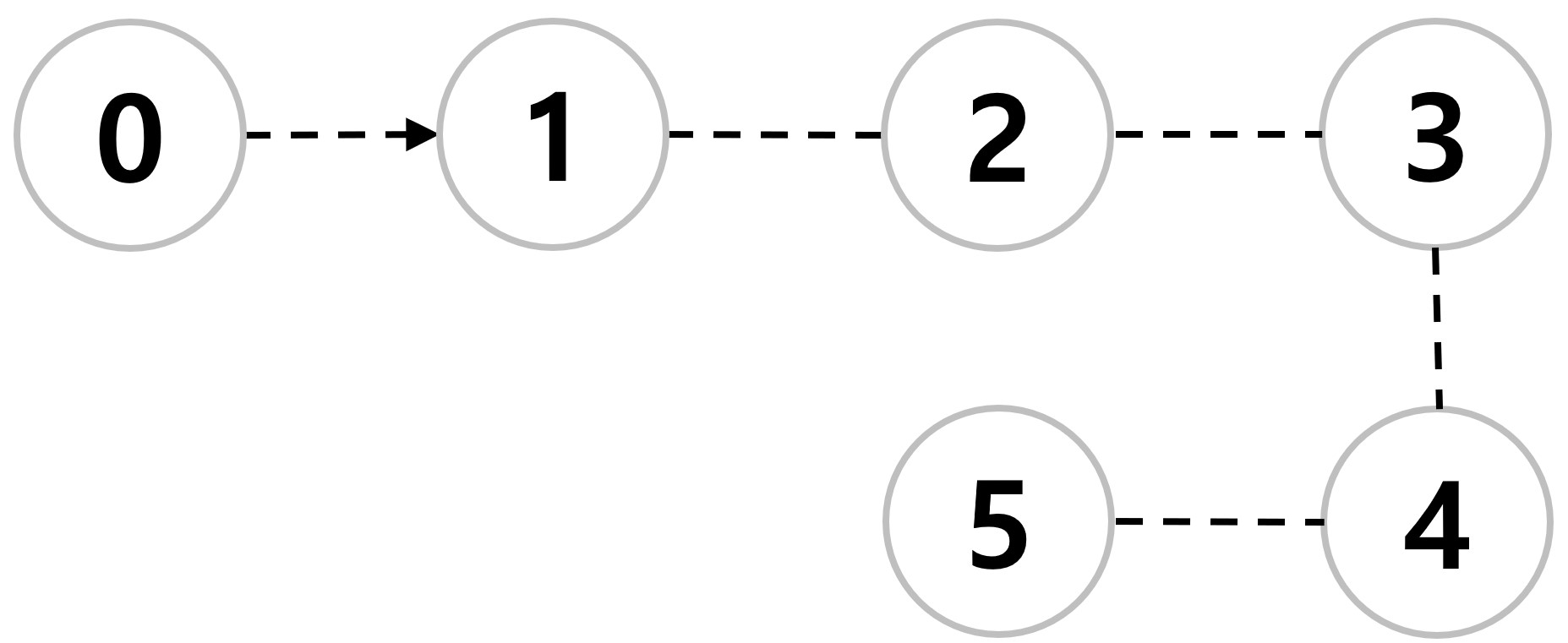}
	\caption{Initial control graph $\bar{\mathcal{G}}(0)$ for Example 1}
	\label{fig:initial-graph-01}
\end{figure}

\begin{figure}
	\centering
	\subfigure[]{\epsfig{figure=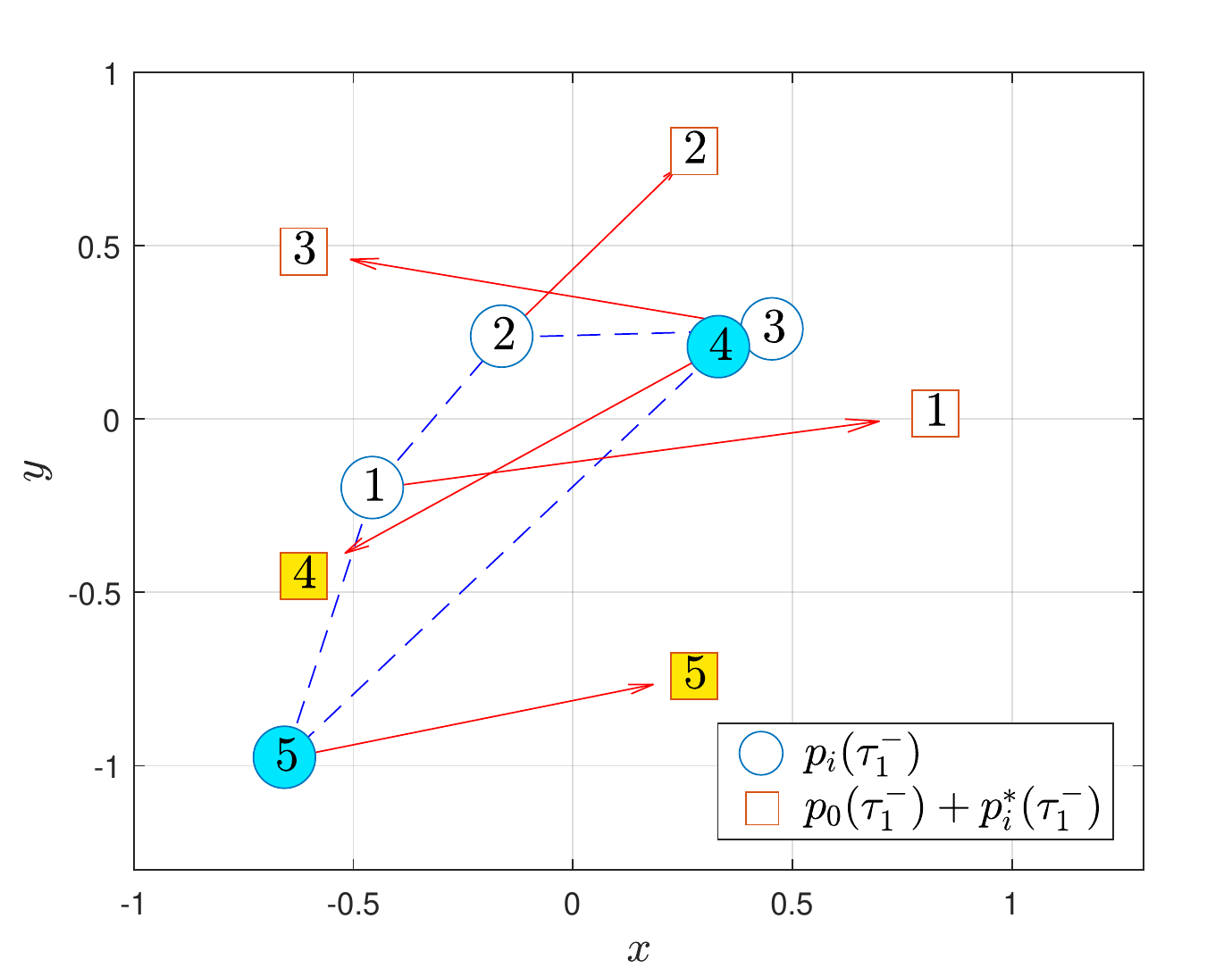, width=5.8cm,height=4.6cm}}
	\subfigure[]{\epsfig{figure=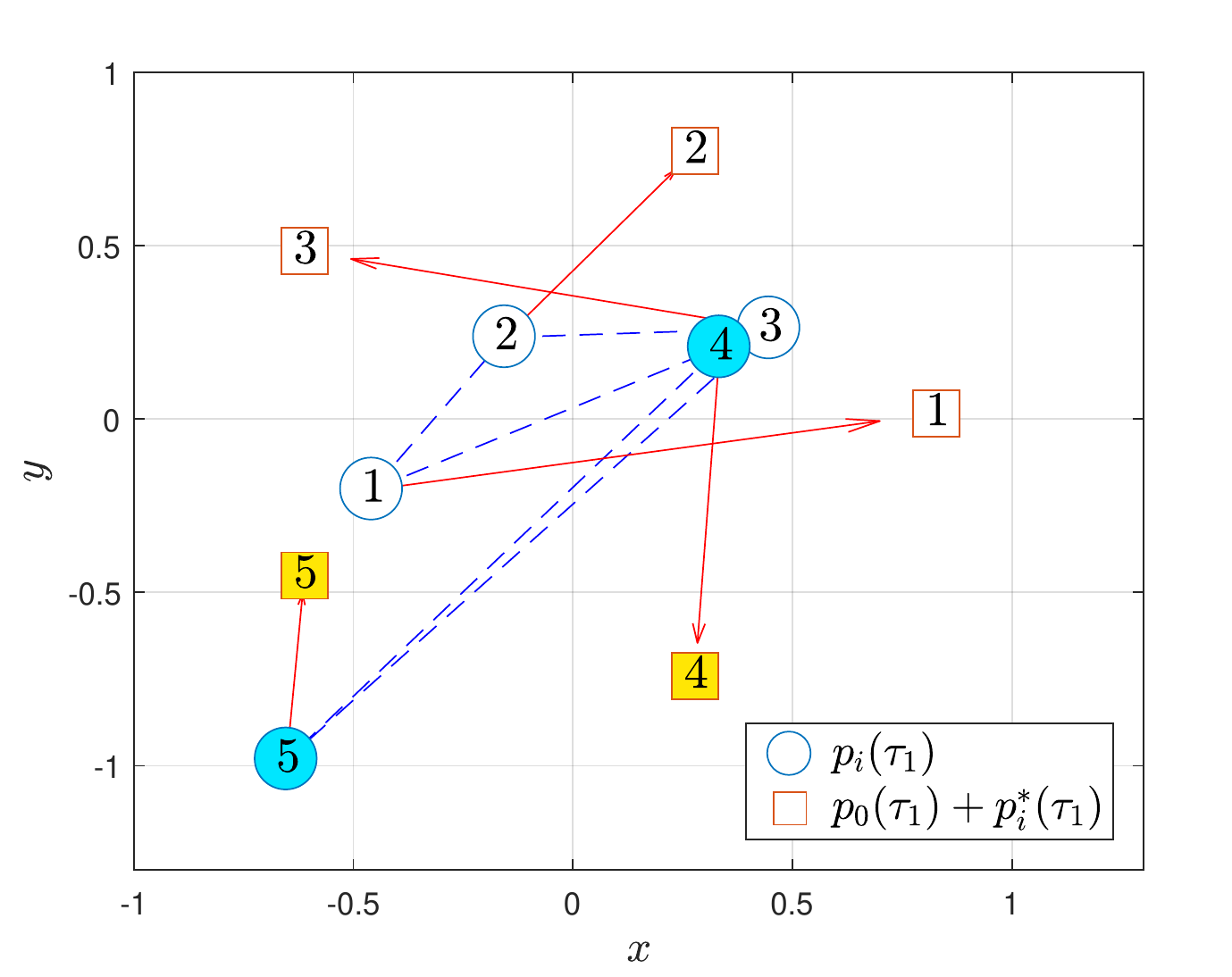, width=5.8cm,height=4.6cm}}
	\subfigure[]{\epsfig{figure=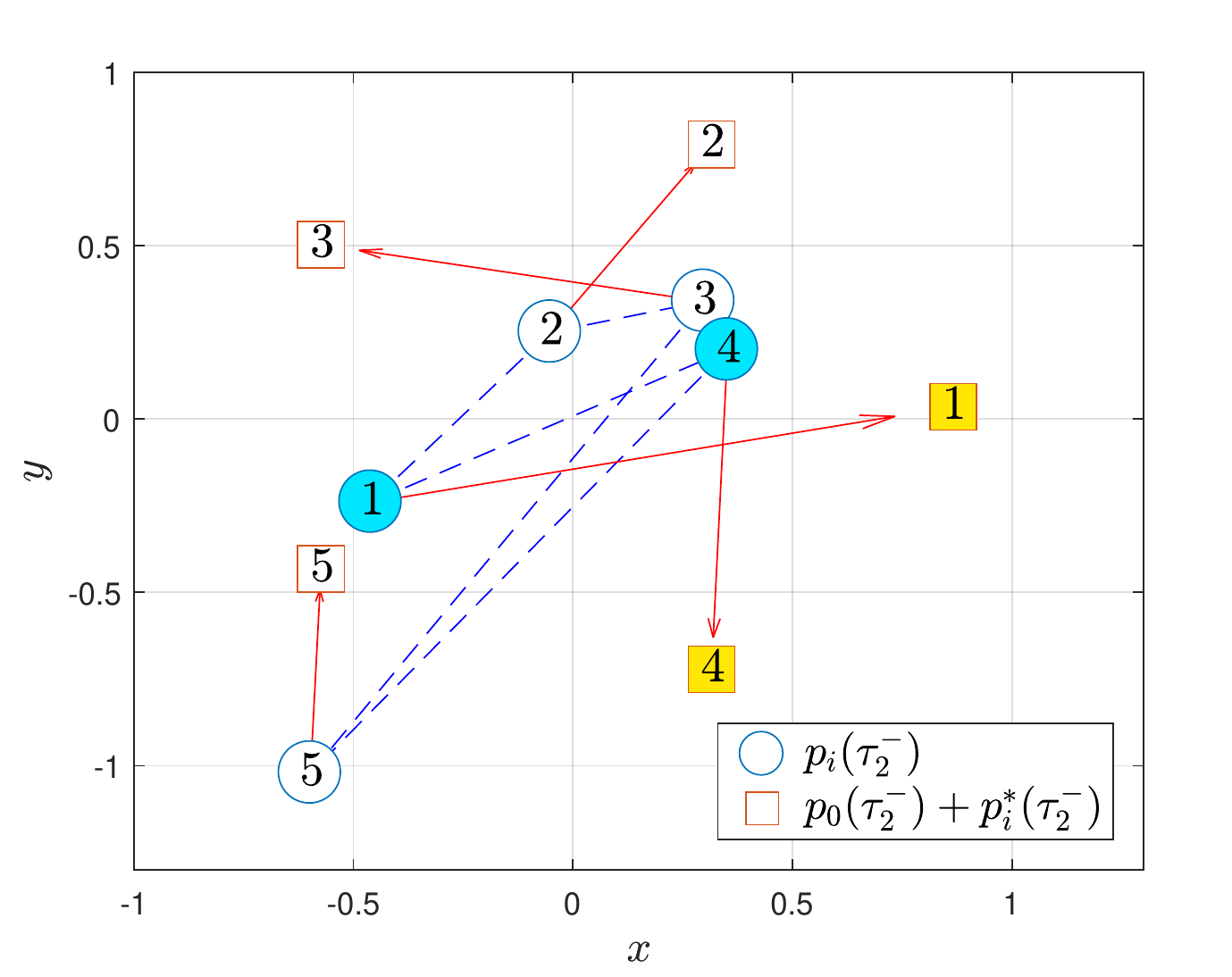, width=5.8cm,height=4.6cm}}
	\subfigure[]{\epsfig{figure=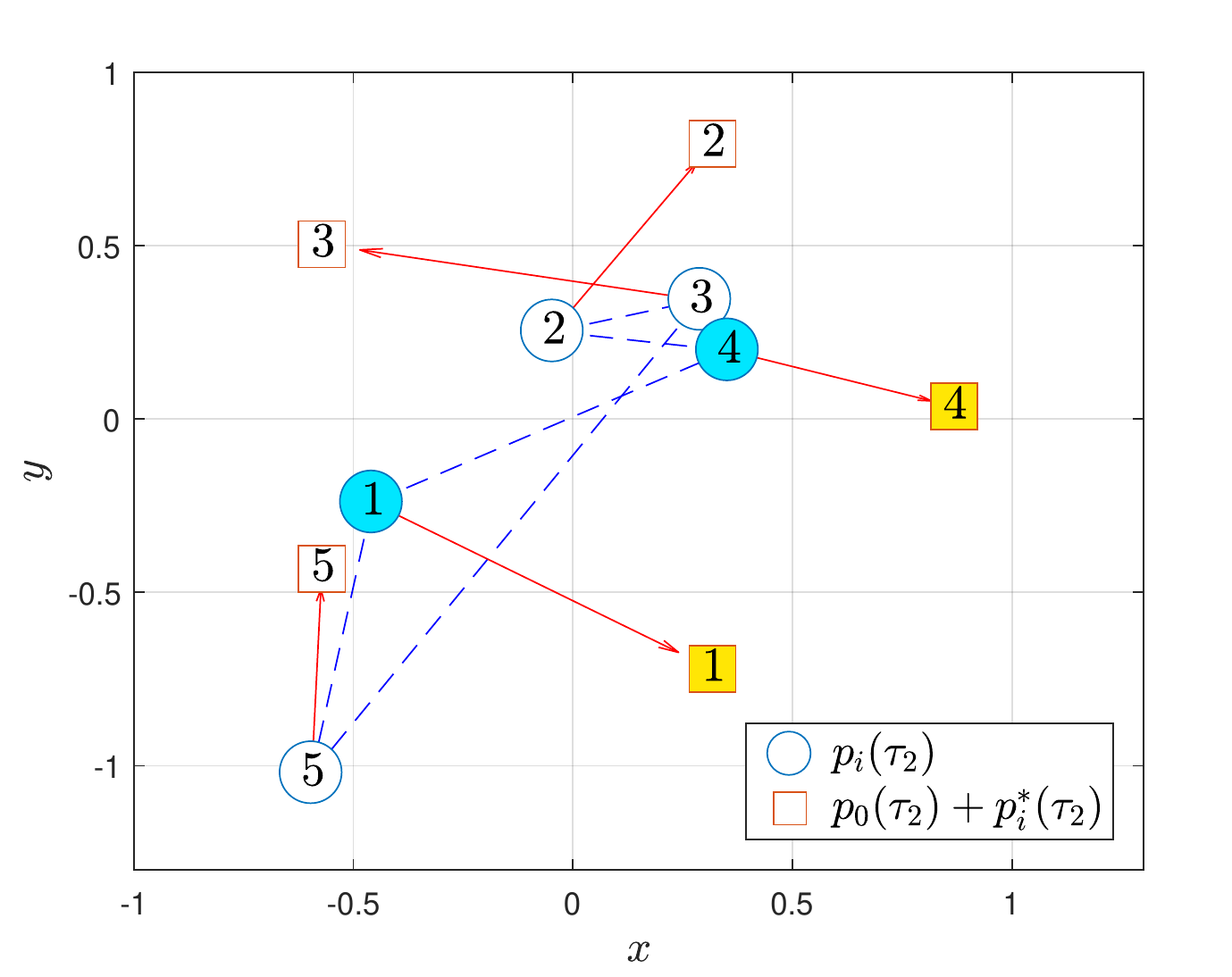, width=5.8cm,height=4.6cm}}
	\caption{Followers' positions and assigned goals for Example 1 at the exchanging instants 
		(a) $t = \tau_{1}^{-}$
		(b) $t = \tau_{1}$
		(c) $t = \tau_{2}^{-}$
		(d) $t = \tau_{2}$.}
	\label{fig:GA-inst-01}
\end{figure}

\begin{figure}
	\centering
	\subfigure[]{\epsfig{figure=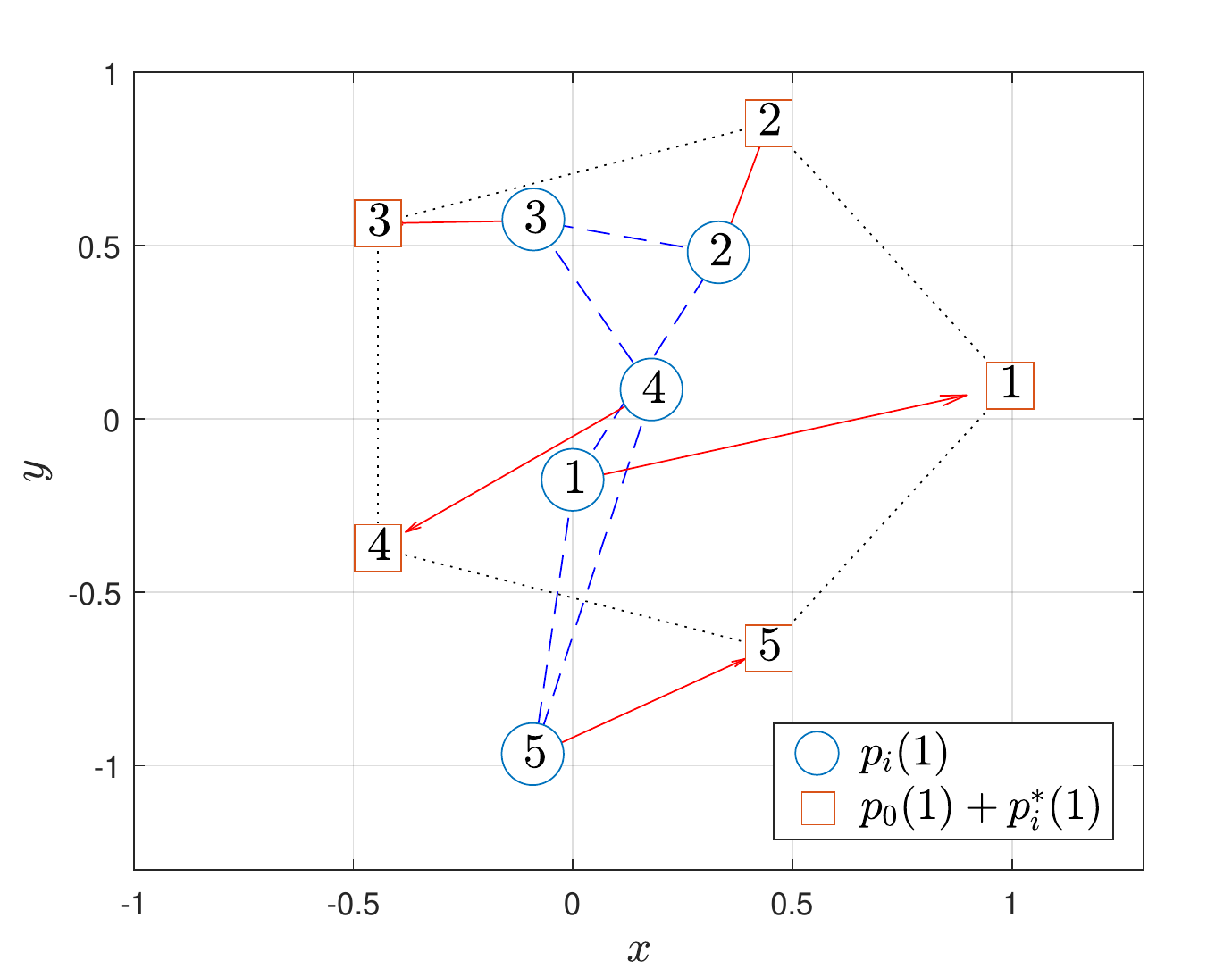, width=5.8cm,height=4.6cm}}
	\subfigure[]{\epsfig{figure=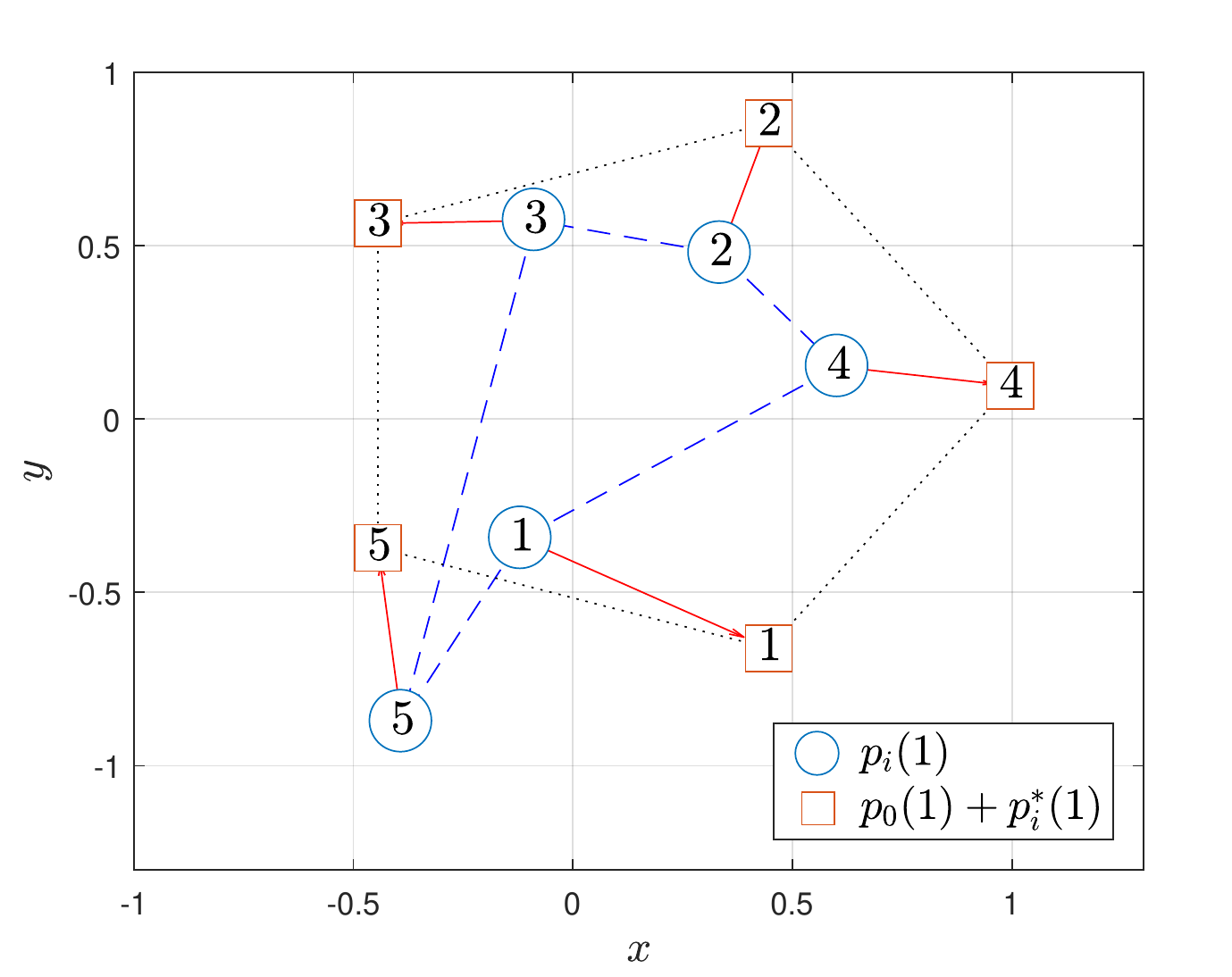, width=5.8cm,height=4.6cm}}
	\caption{Followers' positions and assigned goals at $t = 1$s for Example 1   
		(a) without goal assignment
		(b) with goal assignment.}
	\label{fig:form-01}
\end{figure}

\begin{figure}
	\centering
	\includegraphics[width=5.8cm,height=4.6cm]{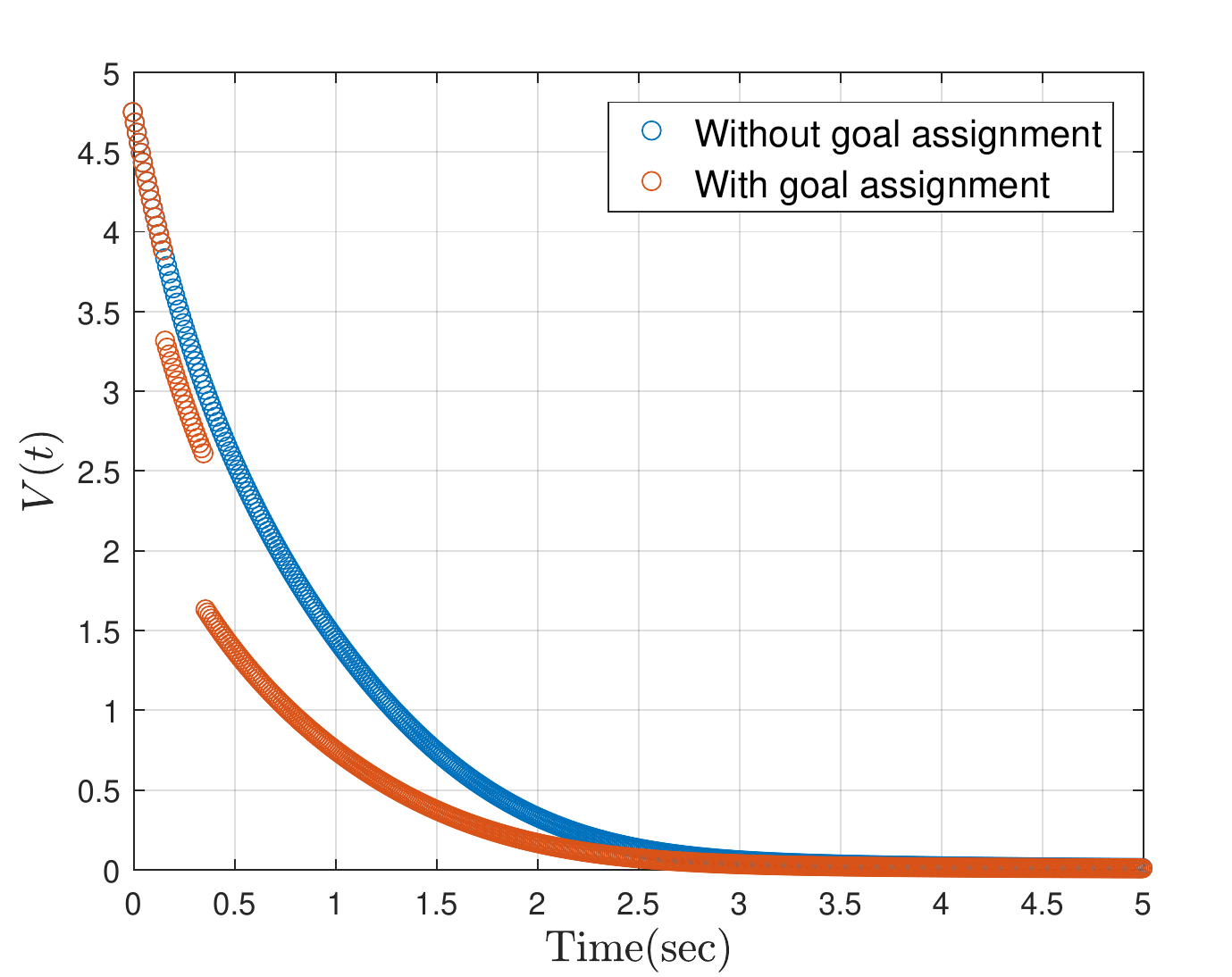}
	\caption{Comparison of the Lyapunov function $V(t)$ for Example 1.}
	\label{fig:lyap-01}
\end{figure}

\section{Simulation results}
\label{sec:simulation}
Two examples are simulated to validate the proposed distributed goal assignment strategy where a numerical example and a multi-quadrotor system are considered in the first and second examples, respectively. 
In this simulation, the formation control results with and without goal assignment will be compared.

\emph{Example 1:} In this example, one leader and five followers in $\mathbb{R}^{2}$ are considered. 
The initial control graph for the leader and followers is given in Fig. \ref{fig:initial-graph-01}. 
The target shape is chosen as a pentagon and the initially assigned goals are $p_{i}^{*}(0) = G_{i}$ for $i=1,\dots,5$
with
\begin{align}
	\nonumber
	\begin{array}{l}
		G_{1}=[0.796,0]^{\top}, ~~~~~~~~~~~~  G_{2}=[0.246, 0.757]^{\top}, ~~~~~
		G_{3}=[-0.644, 0.468]^{\top}, 
		\\	G_{4}=[-0.644, -0.468]^{\top}, ~~ G_{5}=[0.246, -0.757]^{\top}.
	\end{array}
\end{align}
The initial positions and velocities of the followers are randomly chosen. Leader is given by $p_{0}(t) = [0.2t, 0.2 \sin(t/2)]^{\top}$ and the followers are set to communicate with nearby followers within the communication range $R = 2$. 
The controller and estimator gains are selected as $k_{i,1} = 0.5$, $k_{i,2} = 1$, $\gamma_{i,1} = \gamma_{i,2} = 100$, and $\gamma_{i,3} = 20$ where $i=1,\dots,5$.

During the simulation, the goal exchange occurs two times with $\tau_{1} = 0.15$s and $\tau_{2} = 0.35$s and its progress is described in Fig. \ref{fig:GA-inst-01}. In this figure, the blue circles indicate the current positions of the five followers (i.e, $p_{i}(t)$) and the orange rectangles are the sum of the leader positions and goal positions (i.e., $p_{0}(t) + p_{i}^{*}(t)$). 
The red solid arrows in these figures are drawn to show which goal is assigned to each follower and the blue dashed lines represent the edges of the control graph $\bar{\mathcal{G}}_{F}(t)$. Additionally, colored circles and rectangles are for the positions of the followers and their goals to be exchanged, respectively. From Figs. \ref{fig:GA-inst-01}(a) and \ref{fig:GA-inst-01}(a), it is captured that the goals of followers $4$ and $5$ are exchanged at $t = \tau_{1}$. 
It should be emphasized that not only the assigned goals but also the control graph is changed. 
Similarly at $t = \tau_{2}$, the goal exchange between followers $1$ and $4$ happens in Figs. \ref{fig:GA-inst-01}(c) and \ref{fig:GA-inst-01}(d).

The comparison of the formation control results at $t = 1$s is given in Fig. \ref{fig:form-01}. 
As the goals are swapped, the goal positions and the edges regarding to $\bar{\mathcal{G}}_{F}(t)$ in Fig. \ref{fig:form-01}(b) are different from Fig. \ref{fig:form-01}(a).
In these figures, the black dotted lines are drawn to visualize the target shape. 
Note that the formation formed by the followers is relatively closer to the pentagon with goal exchange in Fig. \ref{fig:form-01}(b) than that without goal assignment in Fig. \ref{fig:form-01}(a). 
This reveals that the transient formation control performance is enhanced by the proposed assignment strategy.  
We can see also the performance improvement from Fig. \ref{fig:lyap-01} 
where the blue line is the Lyapunov function $V(t)$ without goal assignment and the orange one is $V(t)$ with goal assignment.
As we discussed in the proof of Theorem \ref{th:final}, there exist discrete jumps of $V(t)$ holding $V(\tau_{g}^{-}) > V(\tau_{g})$ with $g=1,2$ and $V(t)$ with goal assignment is always lower than that without goal assignment. 
From these figures, it appears that the proposed goal assignment strategy provides a performance enhancement of the leader-following formation control.

\emph{Example 2:} As the second example, 3D formation problem of multiple quadrotors is addressed.
A simplified model of the $i$th quadrotor can be modeled as \cite{Labbadi2019}
\begin{align}
	\begin{array}{l}
		\ddot{x}_{i} = \frac{1}{m}(\cos\phi_{i} \sin\theta_{i} \cos \psi_{i} + \sin\phi_{i}\sin\psi_{i})U_{z,i},\\
		\ddot{y}_{i} = \frac{1}{m}(\cos\phi_{i} \sin\theta_{i} \sin \psi_{i} - \sin\phi_{i}\cos\psi_{i})U_{z,i},\\
		\ddot{z}_{i} = -g + \frac{1}{m}\cos\phi_{i} \cos\theta_{i}U_{z,i},\\
		\ddot{\phi}_{i} = a_{1}\dot{\theta}_{i}\dot{\psi}_{i} + b_{1}U_{\phi,i},\\
		\ddot{\theta}_{i} = a_{2}\dot{\phi}_{i}\dot{\psi}_{i} + b_{2}U_{\theta,i},\\
		\ddot{\psi}_{i} = a_{3}\dot{\phi}_{i}\dot{\theta}_{i} + b_{3}U_{\psi,i},\\										
	\end{array}
\label{eq:quadrotor}
\end{align}
where $a_{1} = (I_{yy}-I_{zz})/I_{xx}$, $a_{2} = (I_{zz}-I_{xx})/I_{yy}$, $a_{3} = (I_{xx}-I_{yy})/I_{zz}$, $b_{1} = d/I_{xx}$, $b_{2}=d/I_{yy}$, and $b_{3} = 1/I_{zz}$. The translational and rotational motions of the quadrotors in three directions are assured by $(x_{i},y_{i},z_{i})$ and $(\theta_{i},\phi_{i},\psi_{i})$, respectively.
Here, $m$ is the total mass of the system, $g$ indicates the gravitational acceleration, $I_{xx}$, $I_{yy}$, and $I_{zz}$ denote the rotary inertia, $U_{z,i}$ is the total thrust, $U_{\phi,i}$, $U_{\theta,i}$, and $U_{\psi,i}$ are the moment. The system parameters are given in Table \ref{tab:quadrotor-param}.

\begin{table}[!hbt]
	\setlength{\tabcolsep}{0.5pc}
	\newlength{\digitwidth} \settowidth{\digitwidth}{\rm 0}
	\caption{Quadrotor parameters}
	\label{tab:quadrotor-param}
	\begin{tabularx}{\textwidth}{@{}ccc}
		\toprule
		Parameter & Unit & Value \\
		\hline
		$g$ & m/s$^{2}$ & $9.81$ \\
		$m$ & kg & $0.486$ \\
		$I_{xx}$ & kg$\cdot$m$^{2}$ & $3.827$e-3  \\
		$I_{yy}$ & kg$\cdot$m$^{2}$ & $3.827$e-3  \\
		$I_{zz}$ & kg$\cdot$m$^{2}$ & $7.6566$e-3  \\
		$d$ & m & 0.1  \\
		\bottomrule
	\end{tabularx}
\end{table}

By defining virtual translational inputs $u_{x,i}$, $u_{y,i}$, $u_{z,i}$ as $u_{x,i} = (1/m)(\cos\phi_{i} \sin\theta_{i} $ $\cos \psi_{i} + \sin\phi_{i}\sin\psi_{i})U_{z,i}$, $u_{y,i} = (1/m)(\cos\phi_{i} \sin\theta_{i} \sin \psi_{i} - \sin\phi_{i}\cos\psi_{i})U_{z,i}$, $u_{z,i} = -g + (1/m)$ $\cos\phi_{i} \cos\theta_{i}U_{z,i}$, respectively, the tranlational subsystem becomes a double integrator system as follows: 
\begin{align}
	\begin{array}{l}
		\ddot{x}_{i} = u_{x,i},\\
		\ddot{y}_{i} = u_{y,i},\\
		\ddot{z}_{i} = u_{z,i}.
	\end{array}
	\label{eq:quadrotor-trans}
\end{align}
For \eqref{eq:quadrotor-trans}, the distributed estimator and backstepping controller in Section \ref{sec:estimator-controller} are designed 
and Algorithm \ref{alg:goal} is applied. However, since the actual control input of the quadrotor is $U_{i} = [U_{\phi,i}, U_{\theta,i}, U_{\psi,i}, U_{z,i}]^{\top}$, 
a sliding mode control algorithm is adopted for controlling the attitude subsystem as follows:
\begin{align} 
	\nonumber 
	\begin{array}{l}
	U_{\phi,i} = \frac{1}{b_{1}} (-a_{1}\dot{\theta}_{i}\dot{\psi}_{i} - k_{\phi,i}\mathrm{sign}(s_{\phi,i}) - \lambda_{\phi,i}(\dot{\phi}_{i} - \dot{\phi}_{i,d}) + \ddot{\phi}_{i,d}), \\
U_{\theta,i} = \frac{1}{b_{2}} (-a_{2}\dot{\phi}_{i}\dot{\psi}_{i} - k_{\theta,i}\mathrm{sign}(s_{\theta,i})- \lambda_{\theta,i}(\dot{\theta}_{i} - \dot{\theta}_{i,d}) + \ddot{\theta}_{i,d}),\\ 
U_{\psi,i} = \frac{1}{b_{3}} (-a_{3}\dot{\phi}_{i}\dot{\theta}_{i} - k_{\psi,i}\mathrm{sign}(s_{\psi,i})- \lambda_{\psi,i}(\dot{\psi}_{i} - \dot{\psi}_{i,d}) + \ddot{\psi}_{i,d}), 
	\end{array}
\end{align}
where $k_{\phi,i}$, $k_{\theta,i}$, $k_{\psi,i}$, $\lambda_{\phi,i}$, $\lambda_{\theta,i}$, $\lambda_{\psi,i}$ are control design parameters, $s_{\phi,i}$, $s_{\theta,i}$, and $s_{\psi,i}$ are sliding error surfaces, and $\phi_{i,d}$ and $\theta_{i,d}$ are the desired pith and roll angles, respectively, which are computed from the virtual translational inputs as follows:
\begin{align}
	\begin{array}{l}
		\phi_{i,d} = \arctan \Big(\cos \theta_{d} (\frac{\sin \psi_{i,d} u_{x,i} - \cos \psi_{i,d} u_{y,i}}{u_{z,i}+g}) \Big),\\
		\theta_{i,d} = \arctan \Big(\frac{\cos \psi_{i,d} u_{x,i} + \sin \psi_{i,d} u_{y,i}}{u_{z,i}+g} \Big).
	\end{array}
\end{align}
The desired yaw angle $\psi_{i,d}$ can be chosen freely and thus it is set to zero in our simulation. Similarly, the last control input $U_{z,i}$ is derived as
\begin{align}
	U_{z,i} =  m \sqrt{ u_{x,i}^{2} + u_{y,i}^{2} + (u_{z,i} + g)^{2} }.
\end{align}
For more information on the quadrotor model and the derivation of $\phi_{i,d}$, $\theta_{i,d}$, and $U_{z,i}$, please see \cite{Labbadi2019}. In addition, the sliding mode controllers $U_{\phi,i}$, $U_{\psi,i}$, and $U_{\psi,i}$ for the attitude subsystem are motivated by \cite{Slotine1991}.

\begin{figure}
	\centering
	\includegraphics[width=6cm,height=3.5cm]{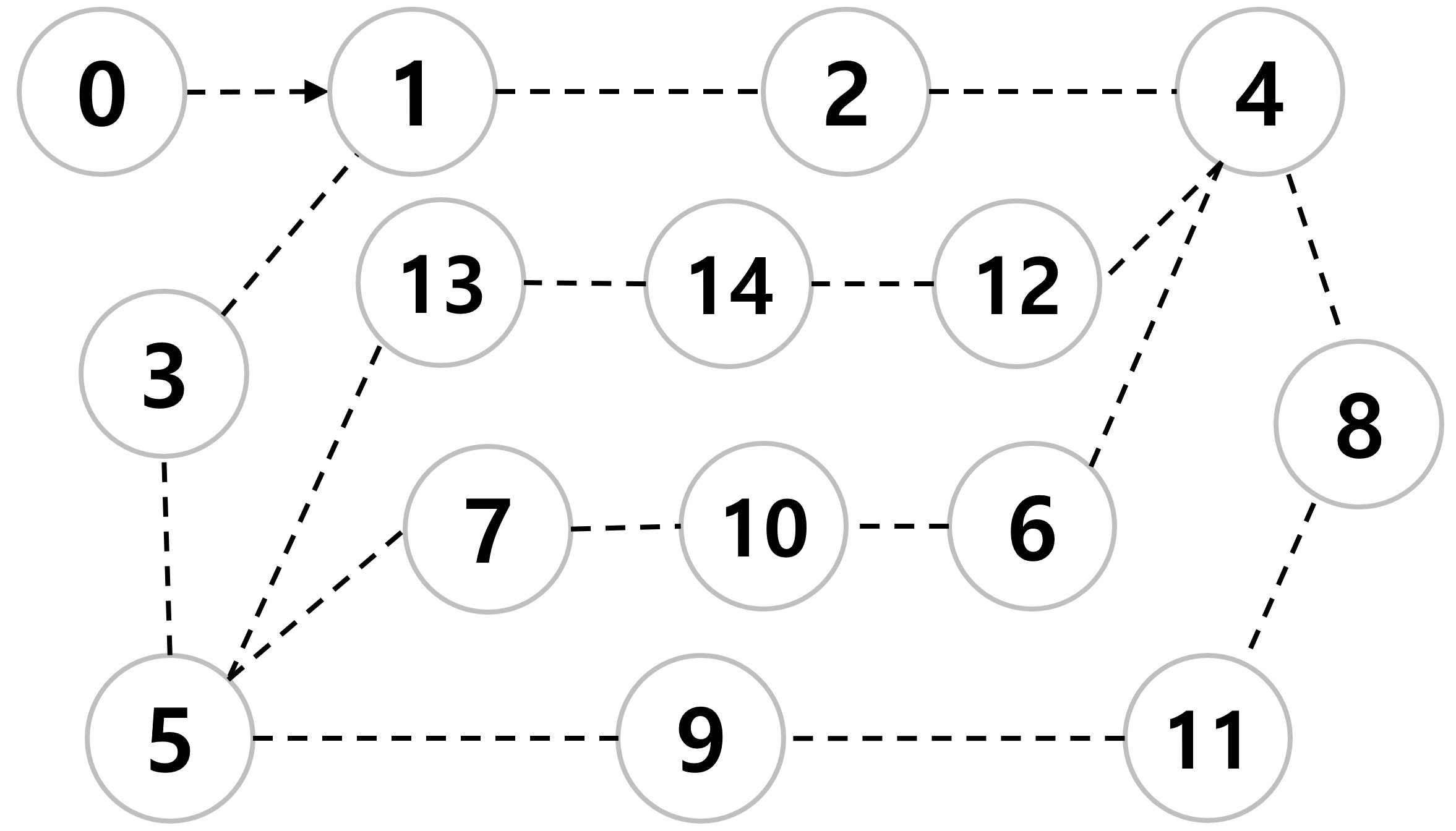}
	\caption{Initial control graph $\bar{\mathcal{G}}(0)$ for Example 2}
	\label{fig:initial-graph-02}
\end{figure}

In this example, a total of $14$ quadrotors are considered as followers and a leader is given by $p_{0} = [10\sin(t/2)~ 10\cos(t/2)~ t+30]^{\top}$.
The system parameters are given in Table \ref{tab:quadrotor-param}.
Goal positions $G_{i}$ are set to form a sphere with respect to the leader and are initially assigned to the quadrotors as $p_{i}^{*}(0) = G_{i}$, $i=1,\dots,14$. We assume that the quadrotors start from the ground and fly to their goal positions with respect to the leader. The goal positions and the initial conditions of the quadrotors are omitted due to the limitation of paper length.
For the distributed leader estimators, backstepping controllers, and the sliding mode attitude controllers, the design parameters are chosen as $R = 20$, $\gamma_{i,1} = \gamma_{i,2} = 100$, $\gamma_{i,3} = 20$, $k_{i,1} = 0.01$, $k_{i,2} = 5$, $\lambda_{\phi,i} = \lambda_{\theta,i} = \lambda_{\psi,i} = 100$, 
$k_{\phi,i} = k_{\theta,i} = k_{\psi,i} = 5$ where $i=1,\dots,14$. 

The initial control graph for the considered multiple quadrotors is given in Fig. \ref{fig:initial-graph-02}.
According to the proposed assignment strategy, there are eight times of goal exchanging with $\tau_{1} = 0.4$s, $\tau_{2}= 0.5$s, $\tau_{3}= 0.8$s, $\tau_{4}= 0.9$s, $\tau_{5}= 1.2$s, $\tau_{6}= 1.3$s, $\tau_{7}= 1.6$s, and $\tau_{8}= 1.9$s. 
Fig. \ref{fig:GA-inst-02} shows the followers' positions and their goals at $t = \tau_{2}$ and $t = \tau_{5}$. 
As illustrated in this figure, two goals of a pair of the followers are swapped which leads to the update of the control graph $\bar{\mathcal{G}}(t)$. 
Fig. \ref{fig:form-02} captures the quadrotors' positions at $t = 5$s with and without goal assignment in xy and xz coordinates.
Despite the same initial conditions and design parameters, the quadrotors with the proposed assignment strategy are closer to their goals than without goal assignment. 
Fig. \ref{fig:lyap-02} illustrates the comparison of the Lyapunov functions $V(t)$ where $V(t)$ 
decreases as the goals are swapped if the proposed assignment algorithm is employed.  
These figures confirm that the proposed assignment strategy is applicable to the formation control of practical multiple quadrotors.

\begin{figure}
	\centering
	\subfigure[]{\epsfig{figure=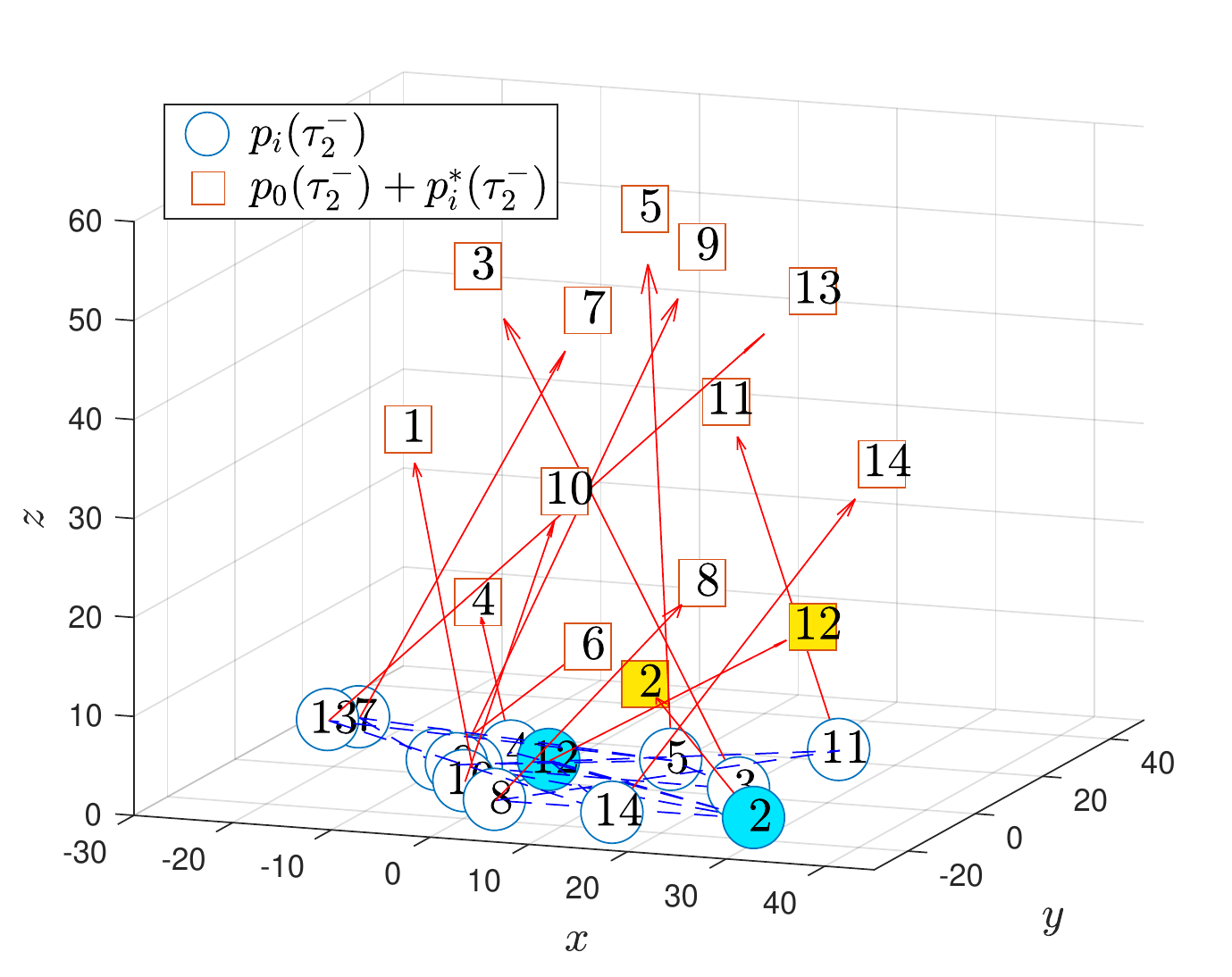, width=5.8cm,height=4.6cm}}
	\subfigure[]{\epsfig{figure=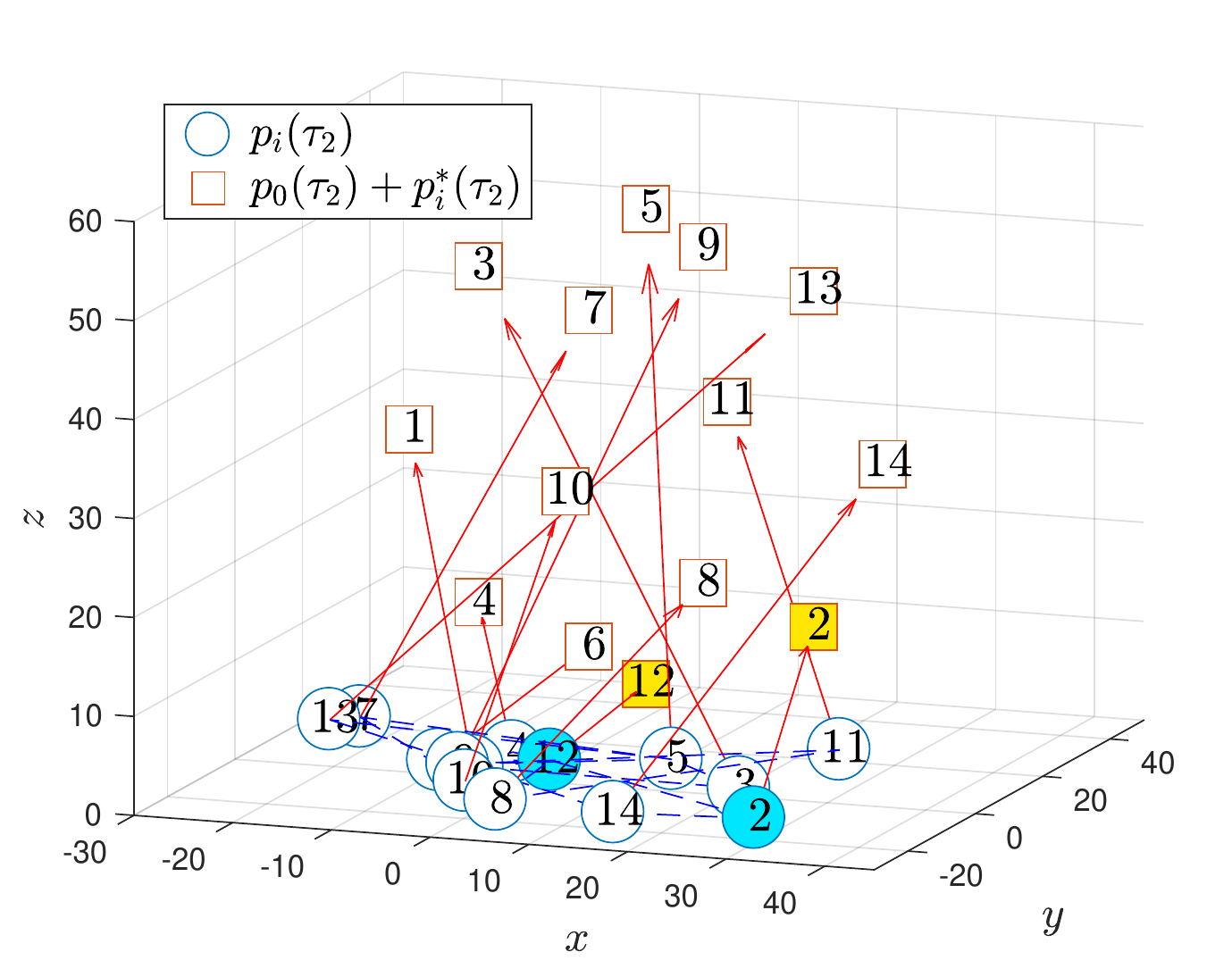, width=5.8cm,height=4.6cm}}
	\subfigure[]{\epsfig{figure=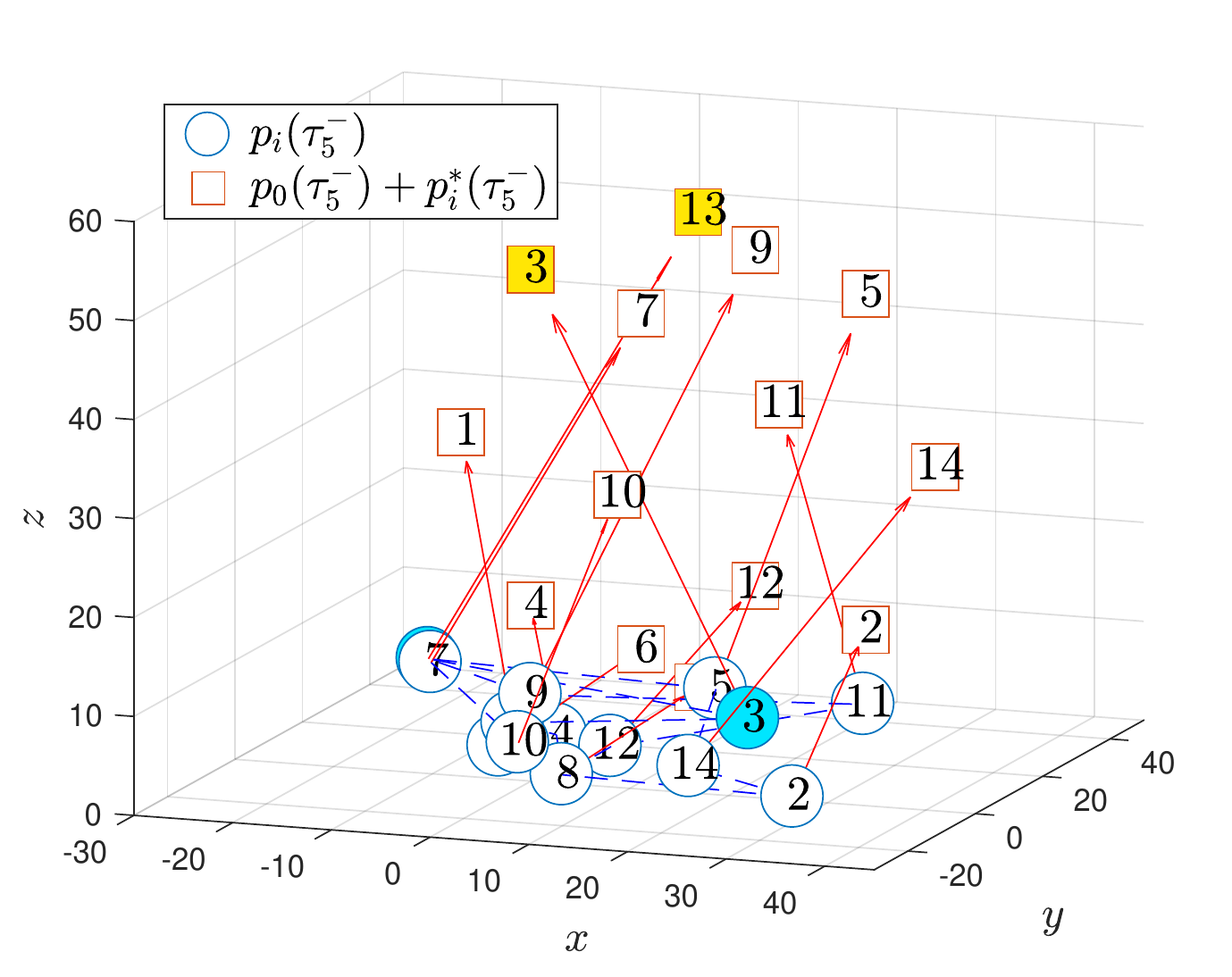, width=5.8cm,height=4.6cm}}
	\subfigure[]{\epsfig{figure=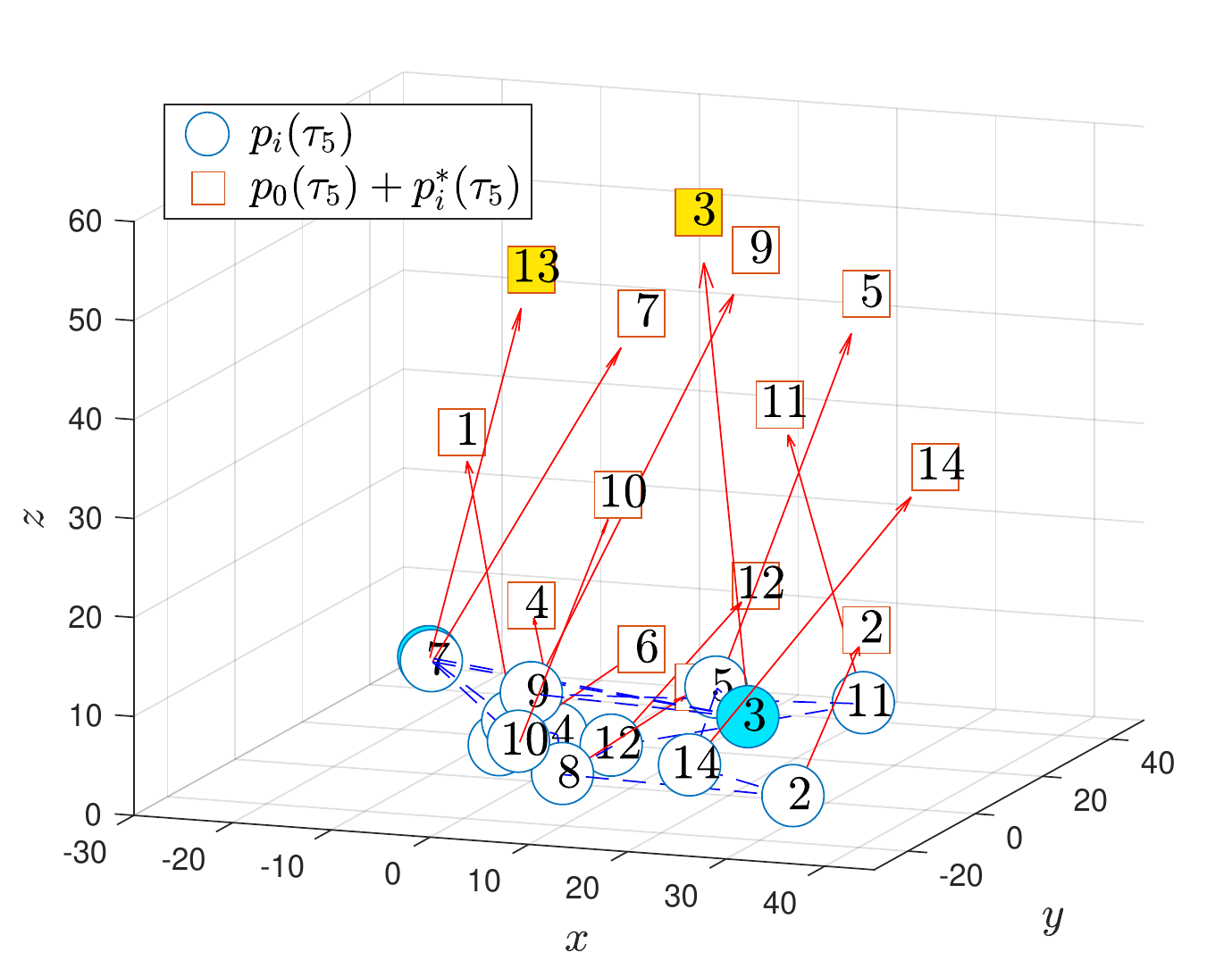, width=5.8cm,height=4.6cm}}
	\caption{Followers' positions and assigned goals for Example 2 at the exchanging instants 
		(a) $t = \tau_{2}^{-}$
		(b) $t = \tau_{2}$
		(c) $t = \tau_{5}^{-}$
		(d) $t = \tau_{5}$.}
	\label{fig:GA-inst-02}
\end{figure}

\begin{figure}
	\centering
	\subfigure[]{\epsfig{figure=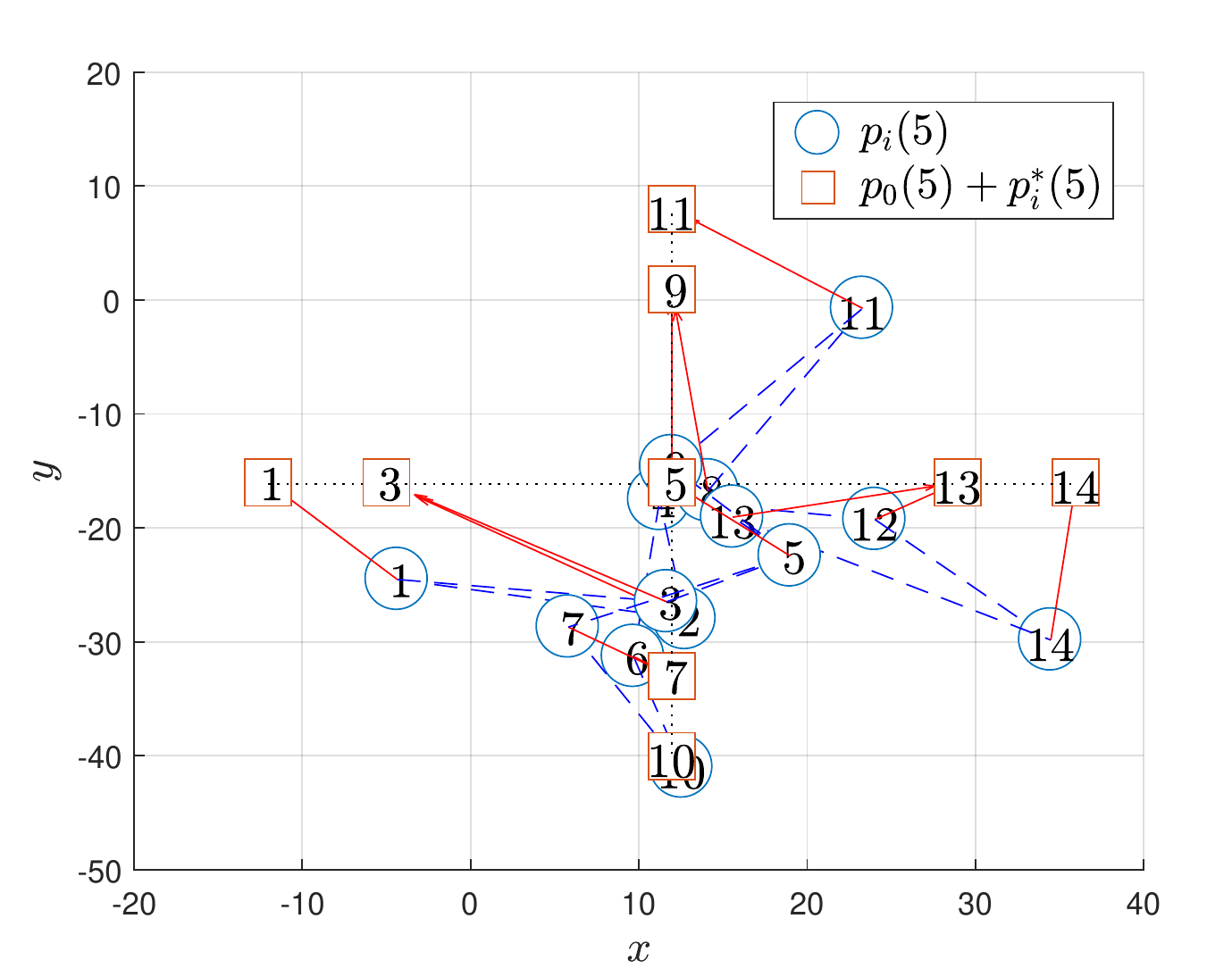, width=5.8cm,height=4.6cm}}
	\subfigure[]{\epsfig{figure=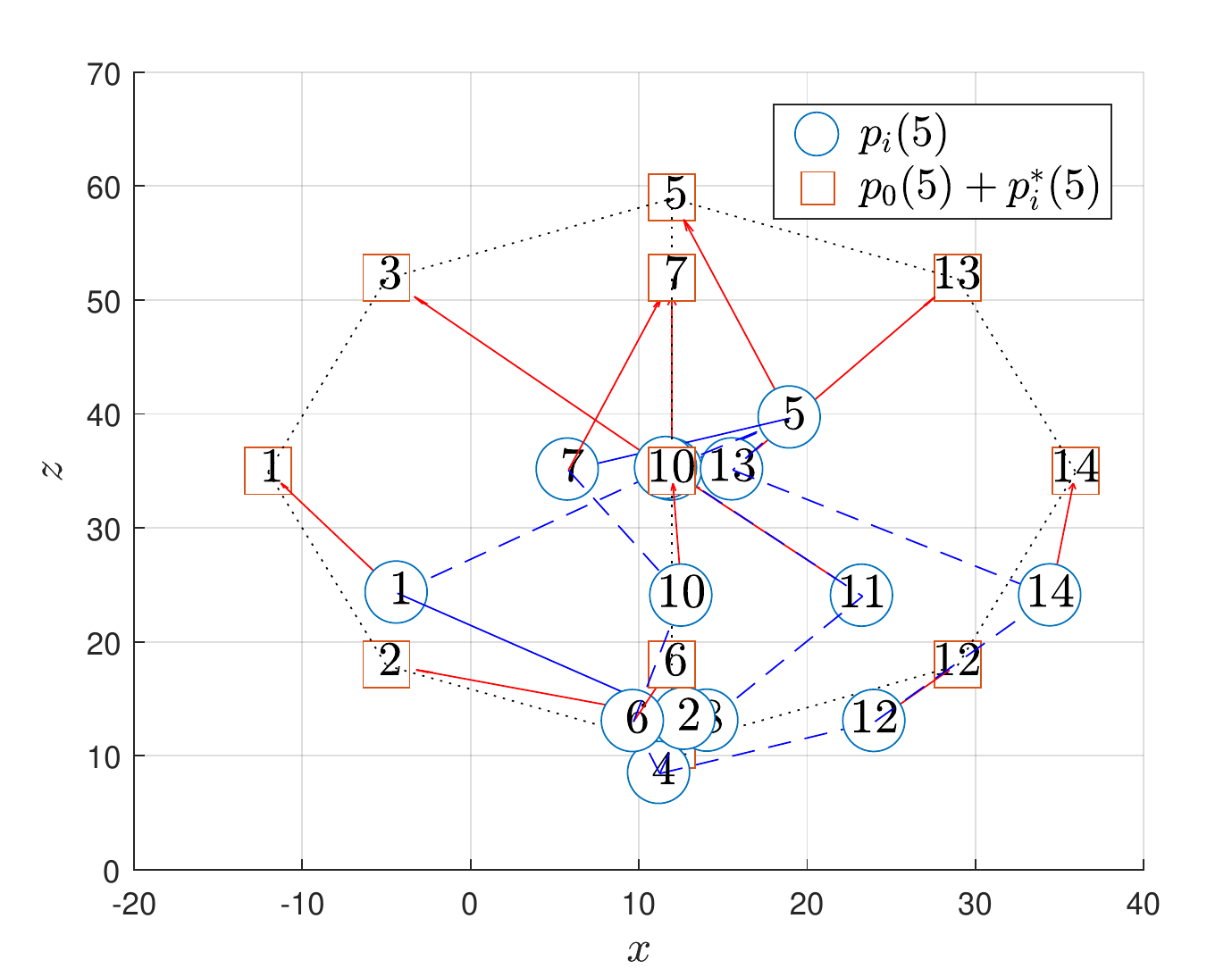, width=5.8cm,height=4.6cm}}
	\subfigure[]{\epsfig{figure=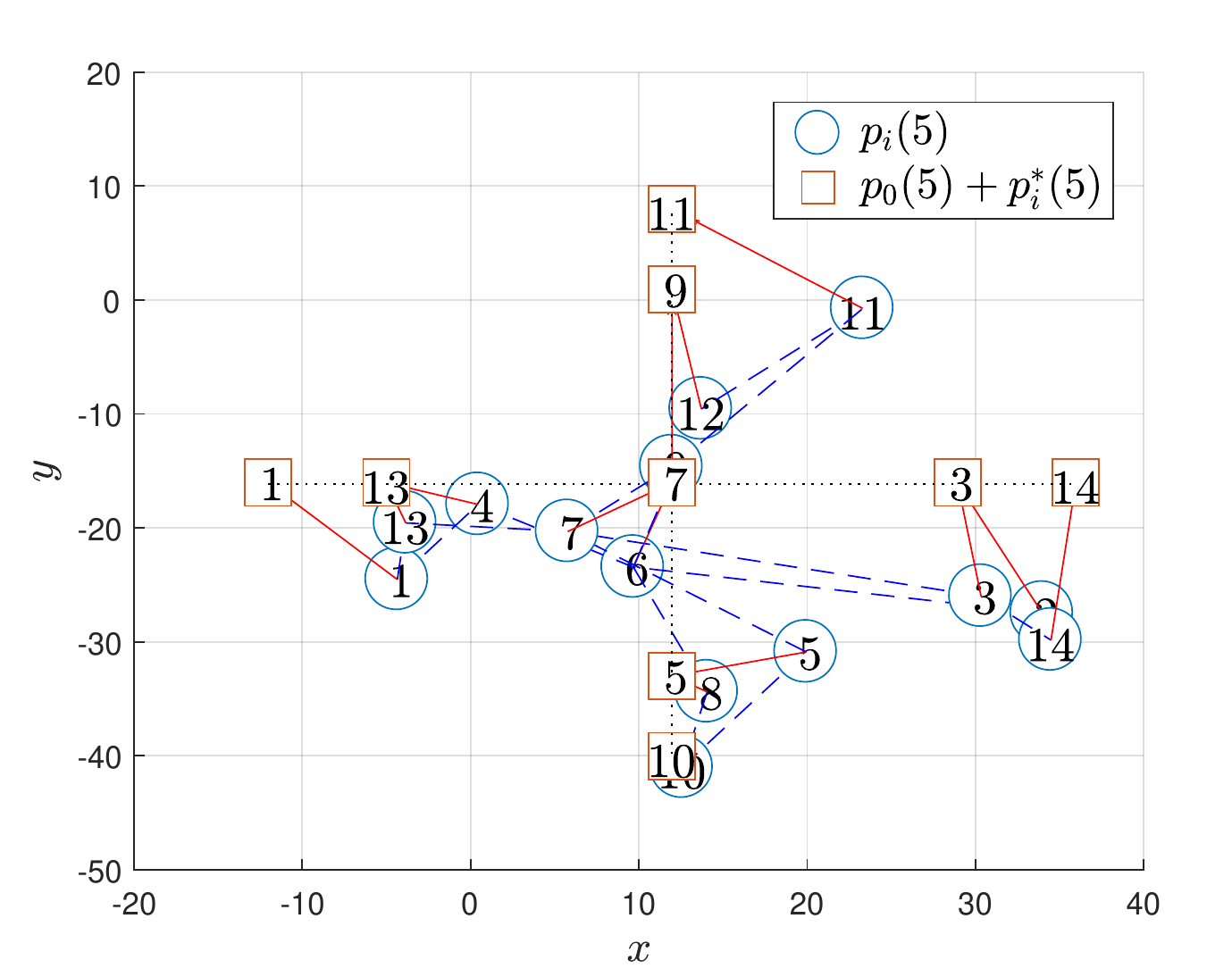, width=5.8cm,height=4.6cm}}	
	\subfigure[]{\epsfig{figure=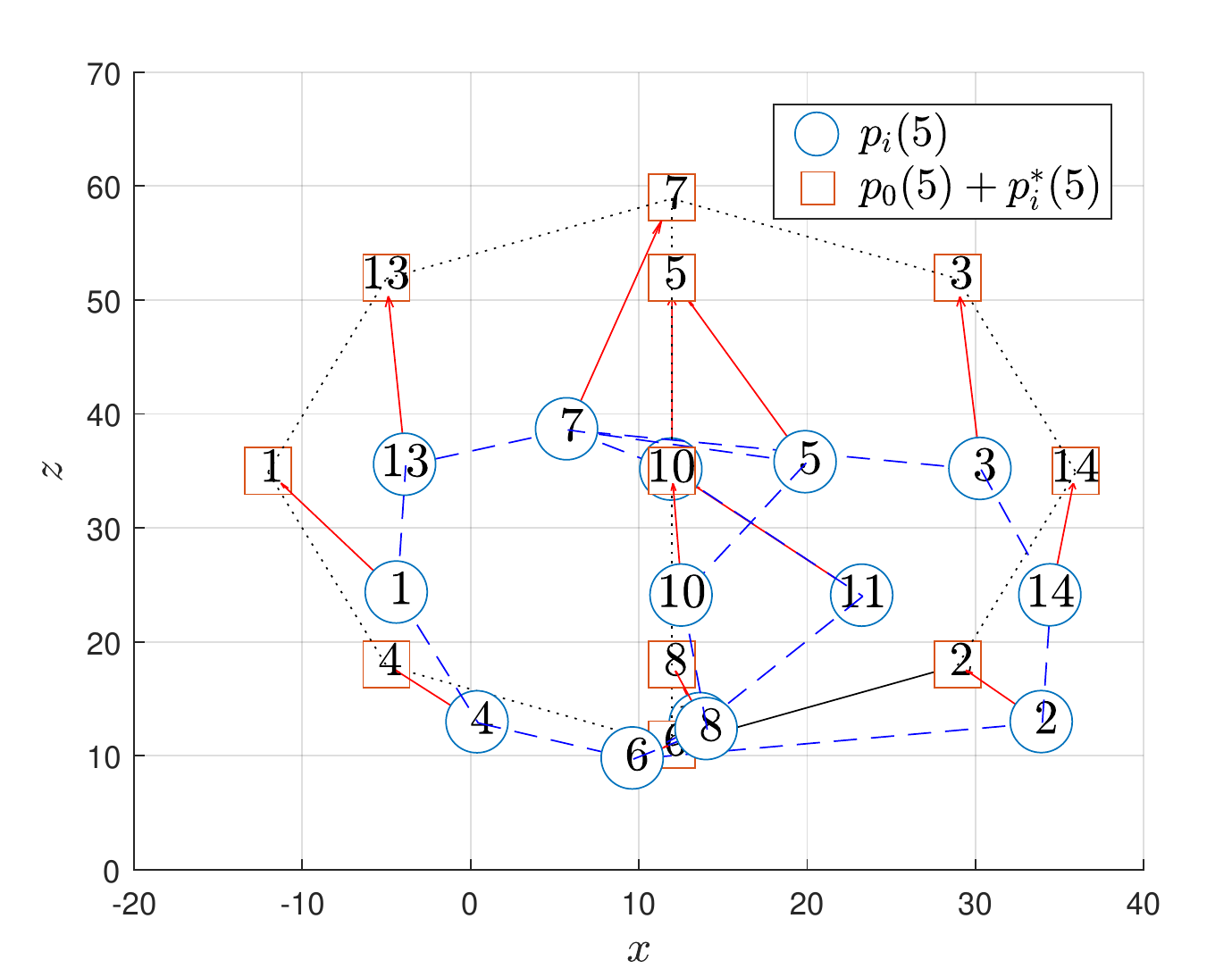, width=5.8cm,height=4.6cm}}		
	\caption{Followers' positions and assigned goals at $t = 5$s  for Example 2 
		(a) without goal assignment in xy coordinates
		(b) without goal assignment in xz coordinates
		(c) with goal assignment in xy coordinates
		(d) with goal assignment in xz coordinates.}
	\label{fig:form-02}
\end{figure}

\begin{figure}
	\centering
	\includegraphics[width=5.8cm,height=4.6cm]{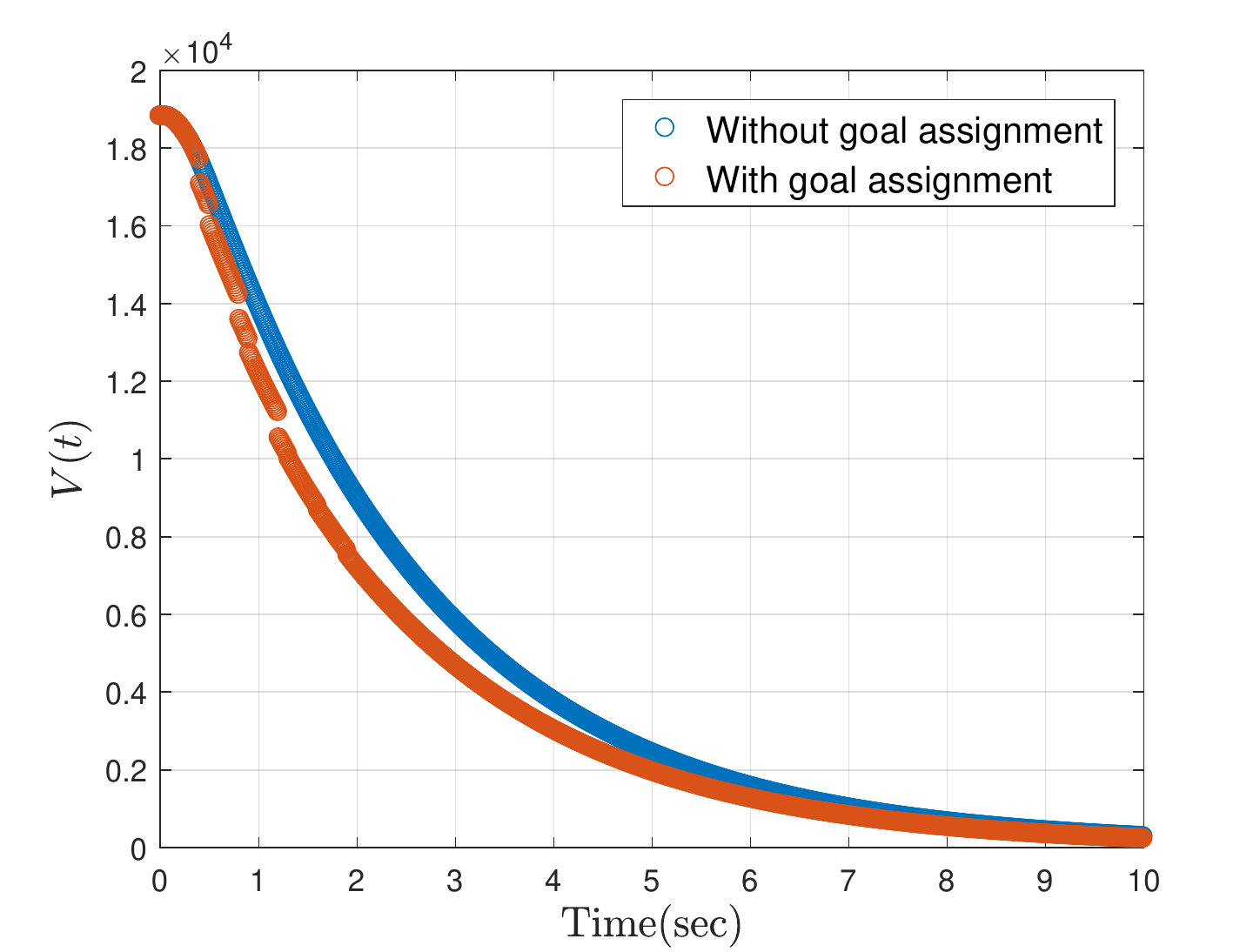}
	\caption{Comparison of the Lyapunov function $V(t)$ for Example 2.}
	\label{fig:lyap-02}
\end{figure}

\section{Conclusion}
\label{sec:conclusion}
A distributed goal assignment and leader-following formation control problem has been solved for second-order multi-agent systems.
Different from the recent goal assignment study, 
the goal positions of followers are given as relative positions with respect to the leader
and the leader information is only available for a small group of the followers.
In order to handle the goal assignment problem in leader-following formation control under distributed communication network,
a distributed estimator and assignment strategy have been presented.
Based on the Lyapunov theory, we have revealed that the assignment strategy improves the control performance
while maintaining the closed-loop stability. 
In addition, the simulation has successfully verified the proposed theoretical result.

\section*{Declaration of Competing Interest}
None.

\section*{Acknowledgment}
This research was supported by the Korea Institute of Science and Technology (KIST) Institutional Program under Grant 2E31581 and by Basic Science Research Program through the National Research Foundation of Korea(NRF) funded by the Ministry of Education(NRF-2021R1A6A3A01086607).

\end{document}